\documentclass[12pt,a4paper]{article}
\usepackage[utf8]{inputenc}
\usepackage[english]{babel}
\usepackage{amsmath}
\usepackage{amsthm}
\usepackage{amsfonts}
\usepackage{amssymb}
\usepackage{geometry}
\usepackage{setspace}
\usepackage[pdftex]{graphicx}
\usepackage[longnamesfirst,authoryear]{natbib}
\usepackage{here}

\newtheorem{Theorem}{Theorem}
\newtheorem{Lemma}{Lemma}
\newtheorem{Proposition}{Proposition}
\newtheorem{Corollary}{Corollary}

\newtheorem{Assumption}{Assumption}
\newtheorem*{Assumption1}{Assumption A.1}
\newtheorem*{Assumption2}{Assumption A.2}
\newtheorem{Example}{Example}

\newtheorem{Remark}{Remark}

\begin{document}

\doublespacing

\author{Takuya Ishihara\footnote{Tohoku University, Graduate School of Economics and Management, email: takuya.ishihara.b7@tohoku.ac.jp. I would like to thank the editor, associate editor, and anonymous referees for their careful reading and comments. I would also like to thank Hidehiko Ichimura, Hiroyuki Kasahara, Masayuki Sawada, Katsumi Shimotsu, and the seminar participants at the University of Tokyo, Kobe University, and Tohoku University. This research was supported by a Grant-in-Aid for JSPS Fellows (20J00900) from the JSPS.}}

\title{Panel Data Quantile Regression for Treatment Effect Models}
\date{\today}
\maketitle

\begin{abstract}
In this study, we develop a novel estimation method for quantile treatment effects (QTE) under rank invariance and rank stationarity assumptions. 
\cite{ishihara2020identification} explores identification of the nonseparable panel data model under these assumptions and proposes a parametric estimation based on the minimum distance method.
However, when the dimensionality of the covariates is large, the minimum distance estimation using this process is computationally demanding.
To overcome this problem, we propose a two-step estimation method based on the quantile regression and minimum distance methods.
We then show the uniform asymptotic properties of our estimator and the validity of the nonparametric bootstrap.
The Monte Carlo studies indicate that our estimator performs well in finite samples.
Finally, we present two empirical illustrations, to estimate the distributional effects of insurance provision on household production and TV watching on child cognitive development.
\end{abstract}

\section{Introduction}

In the literature on program evaluation, it is important to learn about the distributional effects beyond the average effects of the treatment.
Policymakers are more likely to prefer a policy that tends to increase outcomes in the lower tail of the outcome distribution to one that tends to increase outcomes in the middle or upper tail of the outcome distribution.
Such effects can be captured by comparing the quantiles of the treated and control potential outcomes.
The parameter of interest here is the quantile treatment effects (QTE) or the quantile treatment effects on the treated (QTT).
For example, \cite{abadie2002instrumental} estimated the distributional impact of the Job Training Partnership Act (JTPA) program on earnings.
They showed that, for women, the JTPA program had the largest proportional impact at low quantiles.
However, the training impact for men was largest in the upper half of the distribution, with no significant effect on the lower quantiles.
This result could not have been achieved using a mean impact analysis.
Empirical researchers have estimated distributional effects, such as the QTE or QTT, in many areas of empirical economic researches.
For example, \cite{chernozhukov2004effects} estimated the QTE of participation in a 401(k) plan on several measures of wealth; \cite{james2006mean} estimated the QTE of welfare reforms on earnings, transfers, and income; \cite{martincus2010beyond} estimated the QTE of trade promotion activities; and \cite{havnes2015universal} and \cite{kottelenberg2017targeted} estimated the QTT of universal child care.

There is also a rich literature on the identification and estimation of the QTE and QTT in various contexts.
\cite{firpo2007efficient} explored the identification and estimation of the QTE under unconfoundedness.
\cite{abadie2002bootstrap}, \cite{chernozhukov2005iv}, \cite{chernozhukov2006instrumental}, and \cite{frolich2013unconditional} showed how instrumental variables can be used to identify the QTE.
\cite{athey2006identification}, \cite{melly2015changes}, and \cite{callaway2019quantile} provided the identification and estimation results for the QTT in a difference-in-differences (DID) setting by using repeated cross-sections or panel data.
Further, \cite{d2013nonlinear} studied the identification of nonseparable models with continuous treatments using repeated cross sections.

In this study, we use use panel data to develop a novel estimation method for the QTE under rank invariance and rank stationarity assumptions.
We propose a two-step estimator, based on the quantile regression and minimum distance methods.
The rank invariance assumption is used in many nonseparable models, such as those in \cite{matzkin2003nonparametric}, \cite{chernozhukov2005iv}, \cite{d2015identification}, \cite{torgovitsky2015identification}, \cite{feng2020estimation}, and \cite{ishihara2020identification}.
This assumption implies that a scalar unobserved factor determines the potential outcomes across treatment status.
The rank stationarity assumption implies that the conditional distribution of the unobserved factor, given explanatory variables and covariates, does not change over time.
In the literature on nonseparable panel data models, similar assumptions were employed by \cite{athey2006identification}, \cite{hoderlein2012nonparametric}, \cite{graham2012identification}, \cite{d2013nonlinear}, \cite{chernozhukov2013average}, \cite{chernozhukov2015nonparametric}, and \cite{ishihara2020identification}.

\cite{ishihara2020identification} also explores the identification of the nonseparable panel data model under the rank invariance and rank stationarity assumptions.
In this work, the structural function depends on the time period in an arbitrary way and does not require the existence of "stayers" - individuals with the same regressor values in two time periods.
It is important to consider nonlinear time trends when modeling the quantile function.
In this case, additive time trends may be restrictive.
For example, if the quantile function of $Y_t$ is written as $q_t(\tau) = g(\tau) + \mu_t$, the distribution of $Y_t$ is the same across time, up to the location.
However, such an assumption is not valid for many empirical applications.
In contrast, the nonseparable panel data model proposed by \cite{ishihara2020identification} captures nonlinear time effects.

Many nonseparable panel data models require the existence of stayers; this is included in \cite{evdokimov2010identification}, \cite{hoderlein2012nonparametric}, and \cite{chernozhukov2015nonparametric}.
In particular, \cite{evdokimov2010identification} requires the existence of stayers for any value of the treatment variable.
However, many empirically important models do not satisfy this assumption.
For example, in standard DID models, no individuals are treated during both time periods.
The identification approach of \cite{ishihara2020identification} does not require the existence of stayers and allows the support conditions that are employed in standard DID models.

\cite{ishihara2020identification} also proposes a parametric estimation based on the minimum distance method.
However, when the dimensionality of the covariates is large, the minimum distance estimator is computationally demanding.
Hence, when we add many covariates into the model, it is difficult to compute the estimator.
To overcome this problem, we propose a two-step estimation method based on the quantile regression and minimum distance methods.
Using quantile regression, we can obtain an estimator of the QTE by optimizing the objective function over a low-dimensional parameter.
This two-step estimation method is similar to the instrumental variable quantile regression method proposed by \cite{chernozhukov2006instrumental}.

In a DID setting, our model is similar to the changes-in-changes (CIC) model.
\cite{athey2006identification} suggest the CIC model as an alternative to the DID model.
The CIC model allows QTT estimation.
Their model is less restrictive than our model because their approach does not require the rank invariance assumption.
However, their approach does not work when the treatment is a continuous variable or a discrete variable with many different possible values.
There exist many empirical applications in which the treatment variable is continuous, such as when many researchers use panel data to estimate the effect of class size on children's test scores.
Our estimation method, contrary to the CIC model, works when the treatment variable is continuous.

\cite{d2013nonlinear} study the identification of nonseparable models with continuous treatments using repeated cross sections.
They allow for nonlinear time effects by assuming that the structural function $g_t(x,u)$ can be written as $m_t(h(x,u))$, where $m_t$ is a monotonic transformation.
Their study proposes a nonparametric estimation method of the QTT.
However, if we add many covariates into the model, their estimation method does not work because of the curse of dimensionality.

\cite{melly2015changes}, \cite{kottelenberg2017targeted}, and \cite{sawada2019noncompliance} also consider the estimation of the CIC model in the presence of covariates.
\cite{melly2015changes} suggest a flexible semiparametric estimator based on a quantile regression analysis.
They estimate the conditional distribution of outcomes for both treatment and control groups and both periods by using quantile regression, and then apply the changes-in-changes transformations.
Similar to \cite{melly2015changes}, \cite{sawada2019noncompliance} proposes a semiparametric estimator based on distribution regressions.
\cite{kottelenberg2017targeted} rely on the Firpo's (2007) extension to quantiles of the inverse propensity scores method.
However, none of them allow for continuous treatments.

An alternative approach estimates the distributional effects using panel data.
\cite{callaway2019quantile} provide identification and estimation results for the QTT under a straightforward extension of the most common DID assumption.
To identify the QTT, they employ two key assumptions: the distributional difference-in-differences assumption and the copula stability assumption.
The first assumption means that the distribution of the change in potential untreated outcomes does not depend on whether the individual belongs to the treatment or control group.
The second assumption means that the copula between the change in the untreated potential outcomes for the treated group and the initial untreated outcome for the treated group is stable over time.

The rest of the paper is organized as follows.
Section 2 introduces the model and assumptions and demonstrates that our model is nonparametrically identified.
In Section 3, we review the minimum distance estimator, as suggested by \cite{ishihara2020identification}.
We then propose a two-step estimator and show the uniform asymptotic properties of our estimator and the validity of the nonparametric bootstrap.
Section 4 contains the results of several Monte Carlo simulations and we illustrate our estimation method in two empirical settings in Section 5.
The paper concludes in Section 6.
The proofs of the theorems and auxiliary lemmas are provided in the Appendix.

\section{Assumptions and nonparametric identification}

We consider the following potential outcome framework.
The potential outcomes are indexed against the potential values $x$ of the treatment variable $X_{it} \in \mathbb{R}^{d_X}$ and denoted by $Y_{it}(x)$.
We cannot observe $Y_{it}(x)$ directly and the observed outcome is given by $Y_{it}\equiv Y_{it}(X_t)$.
Furthermore, we observe a vector of covariates, $Z_{it}$.
We define $\mathbf{Y}_i \equiv (Y_{i1}, \cdots , Y_{iT})'$, $\mathbf{X}_i \equiv (X_{i1}', \cdots , X_{iT}')'$, $\mathbf{Z}_i \equiv (Z_{i1}', \cdots, Z_{iT}')'$, $W_{it} \equiv (Y_{it},X_{it}',Z_{it}')'$ and $\mathbf{W}_i \equiv (W_{i1}',\cdots,W_{iT}')'$.
Let $\mathcal{X}_t$, $\mathcal{X}_{1:T}$, and $\mathcal{Z}$ denote the supports of $X_{it}$, $\mathbf{X}_i$, and $\mathbf{Z}_i$.

We assume that the potential outcome can be expressed as
\begin{eqnarray}
Y_{it}(x) &=& q_t\left( x,Z_{it},U_{it}(x) \right), \ \ \ \ i = 1, \cdots , n, \ \ t = 1, \cdots, T, \label{potential_outcome}
\end{eqnarray}
where $q_t(x,z_t,\tau)$ is the conditional $\tau$-th quantile of $Y_{it}(x)$ conditional on $\mathbf{Z}_i=(z_1, \cdots, z_T)'$ and $U_{it}(x)$ is uniformly distributed conditional on $\mathbf{Z}_i$.
This implies that the conditional distribution of $Y_{it}(x)$ conditional on $\mathbf{Z}_i$ depends only on $Z_{it}$.
When all covariates are time-invariant, this condition does not restrict the conditional distribution and expression (\ref{potential_outcome}) is known as the Skorohod representation.
Following \cite{chernozhukov2005iv}, we refer to $U_{it}(x)$ as the rank variable.
Additionally, we allow $q_t$ to depend on the time period in an arbitrary manner, similar to the work of \cite{ishihara2020identification}.

First, we impose the rank invariance assumption.

\begin{Assumption}
(i) For all $x$, we have $U_{it}(x) = U_{it}$.
(ii) For all $t$, $U_{it}$ is uniformly distributed on $[0,1]$ conditional on $\mathbf{Z}_i$.
(iii) For all $t$, $\mathbf{x}$, and $\mathbf{z}$, the support of $U_{it}|\mathbf{X}_i=\mathbf{x},\mathbf{Z}_i=\mathbf{z}$ is $[0,1]$. 
\end{Assumption}

Assumption 1 (i) is referred to as the rank invariance assumption.
For example, \cite{matzkin2003nonparametric}, \cite{chernozhukov2005iv}, \cite{d2015identification}, \cite{torgovitsky2015identification}, \cite{feng2020estimation}, and \cite{ishihara2020identification} also employ similar assumptions.
This model is restrictive because the potential outcomes $\{Y_{it}(x)\}_{x \in \mathcal{X}_t}$ are not truly multivariate and, are jointly degenerate.
As discussed in \cite{chernozhukov2005iv}, we can relax the rank invariance assumption to the rank similarity assumption. That is, $U_{it}(x)|\mathbf{X}_i,\mathbf{Z}_i \overset{d}{=} U_{it}(\tilde{x})|\mathbf{X}_i,\mathbf{Z}_i$ for all $x$ and $\tilde{x}$.

Under the rank invariance assumption, the observed outcome can be written as
\begin{eqnarray}
Y_{it} &=& q_t \left( X_{it}, Z_{it}, U_{it} \right),  \ \ \ \ i = 1, \cdots , n, \ \ t = 1, \cdots, T. \label{observed_outcome} 
\end{eqnarray} 
This is the nonseparable model with a scalar unobserved variable and model (\ref{observed_outcome}) is the same as the model proposed by \cite{ishihara2020identification} when there are no covariates.
If $U_{it}$ is independent of $X_{it}$ and $Z_{it}$, then this model is identical with the usual quantile regression model.
However, our model allows for correlation between $U_{it}$ and the treatment variable.
Hence, to achieve point identification, we require additional assumptions.

Next, we impose the rank stationarity assumption.

\begin{Assumption}
For all $t \neq s$, $\mathbf{x}$, and $\mathbf{z}$, the conditional distribution of $U_{it}|\mathbf{X}_i=\mathbf{x},\mathbf{Z}_i =\mathbf{z}$ is the same as that of $U_{is}|\mathbf{X}_i=\mathbf{x}, \mathbf{Z}_i = \mathbf{z}$.
\end{Assumption}

Assumption 2 implies that the rank variable is stationary across the time period.
In the literature on nonseparable panel data models, similar assumptions were employed by \cite{athey2006identification}, \cite{hoderlein2012nonparametric}, \cite{graham2012identification}, \cite{d2013nonlinear}, \cite{chernozhukov2013average}, \cite{chernozhukov2015nonparametric}, and \cite{ishihara2020identification}.
\cite{chernozhukov2013average} referred to Assumption 2 as ``time is randomly assigned'' or ``time is an instrument.''

Assumption 2 can be viewed as a quantile version of the identification condition of the following conventional linear panel data model:
$$
Y_{it} = X_{it}'\alpha + A_i + \epsilon_{it}, \ \ \ \ E[X_{is} \epsilon_{it}]=0 \ \text{for all $t$ and $s$,}
$$
where $A_i$ is a fixed effect and $\epsilon_{it}$ is a time-variant unobserved variable.
Let $\bar{E}[\cdot|\mathbf{X}_i]$ denote the linear projection on $\mathbf{X}_i$, as in \cite{chamberlain1982multivariate}.
\cite{chernozhukov2013average} show that the above equation is satisfied if and only if there is $\tilde{\epsilon}_{it}$ with
$$
Y_{it} = X_{it}'\alpha + \tilde{\epsilon}_{it}, \ \bar{E}[\tilde{\epsilon}_{it}|\mathbf{X}_i]=\bar{E}[\tilde{\epsilon}_{is}|\mathbf{X}_i] \ \text{for all $t$ and $s$.}
$$
In contrast, if the conditional quantile function is linear in $X_{it}$ and there are no covariates, then we can rewrite model (\ref{observed_outcome}) as 
\begin{equation}
Y_{it} = X_{it}'\alpha(\tau) + \epsilon_{it}(\tau), \nonumber
\end{equation}
where $\epsilon_{it}(\tau) = X_{it}'(\alpha(U_{it})-\alpha(\tau))$.
Then, under Assumption 2, $\epsilon_{it}(\tau)$ satisfies $F_{\epsilon_t(\tau)| \mathbf{X}}(0|\mathbf{x}) = F_{\epsilon_s(\tau)| \mathbf{X}}(0|\mathbf{x})$ for all $t\neq s$ and $\mathbf{x}$.
Hence, the rank stationarity assumption can be viewed as a quantile version of the identification condition of the conventional linear panel data model.

Under Assumptions 1 and 2 and additional assumptions in Appendix 1, we can show that $q_t(x,z_t,\tau)$ is nonparametrically identified.
The following proposition is essentially the same as Corollary 1 in \cite{ishihara2020identification}.

\begin{Proposition}
Under Assumptions 1, 2, A.1, and A.2, the conditional quantile function $q_t$ is point identified.
\end{Proposition}

From the proof of Proposition 1, for any $t \neq s$, we have
$$
F_{Y_t|\mathbf{X},\mathbf{Z}}\left( q_t(x_t,z_t,\tau) | \mathbf{x},\mathbf{z} \right) = F_{Y_s|\mathbf{X},\mathbf{Z}}\left( q_s(x_s,z_s,\tau) | \mathbf{x},\mathbf{z} \right),
$$
where $\mathbf{x} = (x_1, \cdots , x_T)'$ and $\mathbf{z} = (z_1, \cdots , z_T)'$.
\cite{ishihara2020identification} demonstrates that this condition provides point identification when the support of $\mathbf{X}_i$ satisfies Assumption A.2.
In Section 3, we propose an estimation method based on this condition.

\begin{Remark}
From the proof of Proposition 1, we can identify $q_t$ from the conditional distribution of $Y_{it}|\mathbf{X}_i$.
This implies that we do not need to observe $(Y_{i1}, \cdots , Y_{iT})$ simultaneously and $q_t$ is identified from repeated cross-sections.
Even when we do not have panel data, we can sometimes observe $(Y_{it},\mathbf{X}_i)$ from repeated cross-sections. 
For example, when $X_{it}$ is the minimum wage at time $t$ in the county where the unit $i$ lives, we can observe $(Y_{it},\mathbf{X}_i)$ from repeated cross-sections if we know the county where the unit $i$ lives.
Although the identification results do not require panel data, we need the existence of panel data in the estimation part.
Hence, in this study, we assume that panel data is obtained.
\end{Remark}

To illustrate our model, we consider the following two examples:

\begin{Example}[The CIC model]
In standard DID models, the support of $(X_{i1},X_{i2})$ becomes $\{(0,0),(0,1)\}$.
Then, $G_i \equiv \mathbf{1} \{ X_{i2} = 1 \}$ denotes an indicator for the treatment group.
In this setting, only individuals in group 1 in period 2 are treated.
Our model is then similar to the CIC model proposed by \cite{athey2006identification}.
Under the assumptions of Proposition 1, we can obtain
\begin{eqnarray}
F_{Y_2(0)|G=1}(y) &=& F_{Y_1|G=1}\left( F_{Y_1|G=0}^{-1}\left( F_{Y_2|G=0}(y) \right) \right), \label{Identification_AI1} \\
F_{Y_2(1)|G=0}(y) &=& F_{Y_1|G=0}\left( F_{Y_1|G=1}^{-1}\left( F_{Y_2|G=1}(y) \right) \right). \label{Identification_AI2}
\end{eqnarray}
\cite{athey2006identification} proved (\ref{Identification_AI1}) without the rank invariance assumption.
Hence, if the target parameter is the QTT, the rank invariance assumption is not required; whereas, if we focus on the QTE, the rank invariance assumption is required.

In this setting, Assumption 2 implies that $U_{i1}|G_i=g,\mathbf{Z}_i=\mathbf{z} \ \overset{d}{=} \ U_{i2}|G_i=g, \mathbf{Z}_i=\mathbf{z}$ for all $g$ and $\mathbf{z}$.
This allows the treatment and control groups to differ in terms of unobservable ability because it does not assume that $U_{it}|G_i=0,\mathbf{Z}_i=\mathbf{z} \ \overset{d}{=} \ U_{it}|G_i=1, \mathbf{Z}_i=\mathbf{z}$.
Hence, the rank stationarity assumption allows the treatment group to contain more high-ability people than the control group, implying that the conditional distribution of $Y_{i2}(0)|G_i = 1, \mathbf{Z}_i=\mathbf{z}$ may be different from that of $Y_{i2}(0)|G_i = 0, \mathbf{Z}_i=\mathbf{z}$.
\end{Example}

\begin{Example}[The TV effect on test scores]
Let $X_{it}$ denote the daily TV watching hours and $Y_{it}$ denote the test score of student $i$ in year $t$.
We assume that $Y_{it}$ can be written as
\[
Y_{it} = q_t(X_{it},Z_{it},U_{it}),
\]
where $Z_{it}$ is a vector of observed characteristics.
We assume that the unobserved factor $U_{it}$ can be decomposed into time-variant and time-invariant parts.
Let $U_{it} = U(A_i,\epsilon_{it})$, where $A_i$ and $\epsilon_{it}$ represent the students' ability and idiosyncratic shocks, respectively.
Here, Assumption 2 is satisfied when we have
\[
\epsilon_{it} | \mathbf{X}_i , \mathbf{Z}_i, A_i \ \overset{d}{=} \ \epsilon_{is} | \mathbf{X}_i , \mathbf{Z}_i, A_i.
\]
Hence, the rank stationarity assumption does not impose any restrictions on the dependence between unobserved ability and TV watching.

In this example, it is important to model nonlinear time effects; for example, if $q_t$ does not change over time, it follows from Assumption 2 that the conditional distribution of $Y_{it}$ is the same across time.
However, this is not plausible because the difficulty level of the test changes over time.
\end{Example}

\section{Estimation and inference}

In this section, we consider the estimation method of the QTE.
First, in Section 3.1, we review the minimum distance method proposed by \cite{ishihara2020identification} and show that the minimum distance estimator does not work when there are many covariates.
Second, in Section 3.2, we propose a two-step estimator based on the quantile regression and minimum distance methods and show that our estimator is computationally convenient.
Finally, in Sections 3.3 and 3.4, we demonstrate the consistency and uniform asymptotic normality of our estimator.

\subsection{The minimum distance estimator}

In this section, for simplicity, we assume that $T=2$ and there are no covariates.
\cite{ishihara2020identification} considers the following parametric model:
\[
q_t(x_t,\tau) = g_t(x_t,\tau ; \theta_0).
\]
The structural functions are parameterized by $\theta \in \Theta \subset \mathbb{R}^{d_{\theta}}$, where $\theta_0 \in \Theta$ is the true parameter.
Then, from Assumption 2, we obtain
\begin{eqnarray}
E\left[ \mathbf{1} \left\{ Y_{i1} \leq g_1(X_{i1},\tau ; \theta_0) \right\} | \mathbf{X}_i \right] = E\left[ \mathbf{1} \left\{ Y_{i2} \leq g_2(X_{i2},\tau ; \theta_0) \right\} | \mathbf{X}_i \right]. \label{conditional_moment_condition}
\end{eqnarray}
Thus, \cite{ishihara2020identification} proposes a minimum distance estimator based on (\ref{conditional_moment_condition}).

Let $\| \cdot \|_{\mu}$ denote the $L_2$-norm with respect to a probability measure $\mu$ with support $[0,1] \times \mathcal{V}$.
The minimum distance estimator $\hat{\theta}$ is then obtained from the following optimization:
\begin{eqnarray}
\hat{\theta} &=& \min_{\theta} \| \hat{D}_{\theta} \|_{\mu}, \label{MD_estimator} \\
\hat{D}_{\theta}(\tau,v) &\equiv & \frac{1}{n} \sum_{i=1}^n \left( \mathbf{1} \left\{ Y_{i1} \leq g_1(X_{i1},\tau ; \theta_0) \right\} - \mathbf{1} \left\{ Y_{i2} \leq g_2(X_{i2},\tau ; \theta) \right\} \right) \omega(\mathbf{X}_i,v), \nonumber
\end{eqnarray}
where $\omega(\mathbf{x},v)$ is a weight function.
Since $\hat{D}_{\theta}(\tau,v)$ is not continuous in $\theta$, the minimum distance estimator requires minimizing the discontinuous objective function over $\theta \in \Theta$.
If the dimension of $\theta$ is large, the optimization (\ref{MD_estimator}) is computationally demanding.
Therefore, adding many covariates into the model makes it difficult to compute the minimum distance estimator.

\subsection{A two-step estimator}

For estimation, we focus on the following linear-in-parameter model:
\begin{equation}
q_t(x_t,z_t,\tau) = x_t'\alpha(\tau) + z_t'\beta_t(\tau). \label{Model}
\end{equation}
The observed outcome is then written as
$$
Y_{it} = X_{it}'\alpha(U_{it})+Z_{it}'\beta_t(U_{it}), \ \ \ U_{it}|\mathbf{Z}_i \sim U(0,1),
$$
where $Z_{it}$ contains a constant term.
Hereafter, we set $\mathbf{Z}_i \in \mathbb{R}^{d_z}$ as a vector of all the variables of $Z_{i1}, \cdots, Z_{iT}$.
For example, if all covariates are time-invariant, we have $Z_{i1} = \cdots = Z_{iT} = \mathbf{Z}_i$.
We assume that $\mathcal{X}_{t}$ and $\mathcal{Z}$ are bounded.
In this model, we have $\partial q_t(x,z,\tau) / \partial x = \alpha(\tau)$; hence, our target parameter is $\alpha(\tau)$.
Because $\beta_t(\tau)$ depends on the time period, this model captures nonlinear time effects.
This model is similar to the IV quantile regression model proposed by \cite{chernozhukov2006instrumental}.

Using Proposition 1, we can identify $\alpha(\tau)$ and $\beta_t(\tau)$ using the following conditions:
\begin{eqnarray}
F_{Y_t|\mathbf{X},\mathbf{Z}}(x_t'\alpha(\tau)+z_t'\beta_t(\tau)|\mathbf{x},\mathbf{z}) &=& F_{Y_s|\mathbf{X},\mathbf{Z}}(x_s'\alpha(\tau)+z_s'\beta_s(\tau)|\mathbf{x},\mathbf{z}) \label{moment_1} \\
F_{Y_t-X_t'\alpha(\tau)|\mathbf{Z}}(z_t'\beta_t(\tau)|\mathbf{z}) &=& \tau, \label{moment_2}
\end{eqnarray}
where $\mathbf{x} \equiv (x_1, \cdots ,x_T)'$ and $\mathbf{z} \equiv (z_1, \cdots ,z_T)'$, respectively.
Similar to (\ref{MD_estimator}), we can construct a minimum distance estimator using (\ref{moment_1}) and (\ref{moment_2}).
However, if the dimensionality of covariates is high, the minimum distance approach cannot be directly applied because the minimum distance estimator is computationally demanding.

We propose the following two-step estimator based on the quantile regression and minimum distance methods.
Fix $\tau \in (0,1)$.
In the first step, we define $\tilde{\beta}_t(a,\tau)$ as
\begin{eqnarray}
\tilde{\beta}_t(a,\tau) & \equiv & \text{arg} \min_{b_t \in \mathcal{B}_t} \frac{1}{n} \sum_{i=1}^n \rho_{\tau}\left( Y_{it} - X_{it}'a - Z_{it}'b_t \right) \nonumber \\
&=& \text{arg} \min_{b_t \in \mathcal{B}_t} \frac{1}{n} \sum_{i=1}^n R_{\tau}(W_{it};a,b_t),  \label{first step estimator}
\end{eqnarray}
where $\rho_{\tau}(u) \equiv (\tau - \mathbf{1}\{u<0\})u$, $\mathcal{B}_t$ is the parameter space of $\beta_t(\tau)$, and $R_{\tau}(W_{it};a,b_t) \equiv \rho_{\tau}\left( Y_{it} - X_{it}'a - Z_{it}'b_t \right)$.
This is an ordinary quantile regression of $Y_{it}-X_{it}'a$ on $Z_{it}$.
Then, from (\ref{moment_2}), $\tilde{\beta}_t \left( \alpha(\tau),\tau \right)$ becomes a consistent estimator of $\beta_t(\tau)$.

In the second step, we construct an estimator of $\alpha(\tau)$ using the minimum distance approach.
We define 
\begin{eqnarray}
g_t(\mathbf{W}_i;a,b,v) &\equiv & \left( \mathbf{1}\{Y_{it} \leq X_{it}'a + Z_{it}'b_t\} - \frac{1}{T} \sum_{s=1}^T \mathbf{1}\{Y_{is} \leq X_{is}'a + Z_{is}'b_s\} \right) \omega(\mathbf{X}_i,\mathbf{Z}_i,v), \nonumber \\
\omega(\mathbf{X}_i,\mathbf{Z}_i,v) &\equiv & \exp\left(v_{\mathbf{x}}'\tilde{\mathbf{X}}_i + v_{\mathbf{z}}'\tilde{\mathbf{Z}}_{i} \right), \nonumber
\end{eqnarray}
where $b = (b_1', \cdots , b_T')'$, $v = (v_{\mathbf{x}}', v_{\mathbf{z}}')'$, and $\tilde{\mathbf{X}}_i$ and $\tilde{\mathbf{Z}}_{i}$ are standardized versions of $\mathbf{X}_i$ and $\mathbf{Z}_i$, where each component has a mean of 0 and a standard deviation of 1.
It follows from (\ref{moment_1}) that we have
\begin{equation}
E\left[ g_t(\mathbf{W}_i;\alpha(\tau),\beta(\tau),v) \right] = 0 \ \ \text{for all $t$ and $v$,} \label{moment_second_step}
\end{equation}
where $\beta(\tau) \equiv (\beta_1(\tau)', \cdots, \beta_T(\tau)')'$.
As shown in \cite{stinchcombe1998consistent}, if (\ref{moment_second_step}) holds for all $t$ and $v \in \mathcal{V} \equiv [-0.5,0.5]^{d_{X}\cdot T + d_{Z}}$, the conditional moment condition (\ref{moment_1}) is satisfied.
Let $\|\cdot\|_{L_2}$ be the $L_2$-norm over a compact set $\mathcal{V}$; that is, $\|f(v)\|_{L_2}^2 = \int_{\mathcal{V}} f(v)^2 dv$.
Using this norm, we obtain the following estimator of $\alpha(\tau)$:
\begin{eqnarray}
\hat{\alpha}(\tau) &\equiv & \text{arg} \min_{a \in \mathcal{A}} \frac{1}{T} \sum_{t=1}^T \left\| \frac{1}{n} \sum_{i=1}^n g_t(\mathbf{W}_i;a,\tilde{\beta}(a,\tau),v) \right\|_{L_2}^2 \nonumber \\
& = & \text{arg} \min_{a \in \mathcal{A}} \frac{1}{T} \sum_{t=1}^T \left\| \hat{D}^t_n(v;a,\tilde{\beta}(a,\tau)) \right\|_{L_2}^2, \label{Estimator alpha}
\end{eqnarray}
where $\tilde{\beta}(a,\tau) \equiv (\tilde{\beta}_1(a,\tau)', \cdots , \tilde{\beta}_T(a,\tau))'$, $\hat{D}^t_n(v;a,b)\equiv \frac{1}{n} \sum_{i=1}^n g_t(\mathbf{W}_i;a,b,v)$, and $\mathcal{A}$ is the parameter space of $\alpha(\tau)$.
Finally, we estimate $\beta_t(\tau)$ by $\hat{\beta}_t(\tau) \equiv \tilde{\beta}_t(\hat{\alpha}_t(\tau),\tau)$.

We briefly explain our two-step estimation method.
As discussed above, because the covariates are independent of $U_{it}$, $\tilde{\beta}_t \left( \alpha(\tau),\tau \right)$ becomes a consistent estimator of $\beta_t(\tau)$.
Using this result, under regularity conditions, we obtain
\[
\hat{D}_n^t \left( v ; \alpha(\tau), \tilde{\beta}(\alpha(\tau),\tau) \right) \rightarrow_p E[g_t(\mathbf{W}_i; \alpha(\tau), \beta(\tau),v)] = 0,
\]
which implies that the objective function of (\ref{Estimator alpha}) converges to zero for $a = \alpha(\tau)$.
Hence, we expect to obtain a consistent estimator of $\alpha(\tau)$ by minimizing (\ref{Estimator alpha}).

In practice, we can implement this estimation procedure as follows:
\begin{enumerate}
\item For fixed $\tau$, we run the ordinary $\tau$-th quantile regression of $Y_{it}-X_{it}'a$ on $Z_{it}$ and calculate $\tilde{\beta}_t(a,\tau)$ as a function of $a$.
\item We approximate $\| \hat{D}^t_n(v;a,\tilde{\beta}(a,\tau)) \|_{L_2}^2$ using a numerical integration method; that is, we approximate the objective function as $\frac{1}{J} \sum_{j = 1}^J \hat{D}^t_n(v_j;a,\tilde{\beta}(a,\tau))^2$ for an appropriate sequence $\{v_j\}_{j=1}^J$.
\item Minimize $\frac{1}{T} \sum_{t=1}^T \| \hat{D}^t_n(v;a,\tilde{\beta}(a,\tau)) \|_{L_2}^2$ over $a \in \mathcal{A}$ and obtain $\hat{\alpha}(\tau)$.
The estimate $\hat{\beta}_t(\tau)$ is given by $\tilde{\beta}_t(\hat{\alpha}(\tau),\tau)$.
\end{enumerate}

Our estimator is similar to that proposed by \cite{chernozhukov2006instrumental}.
They consider the IV quantile regression for heterogeneous treatment effect models and simultaneous equation models with nonadditive errors.
Similarly, our estimator is attractive from a computational point of view.
As ordinary quantile regressions are obtained by convex optimization, our first step estimation (\ref{first step estimator}) is computationally convenient.
Our second step estimation (\ref{Estimator alpha}) requires non-convex optimization; hence, it seems to be computationally demanding.
However, we can obtain (\ref{Estimator alpha}) by optimizing the objective function over the $\alpha$ parameter (typically one-dimensional).
This fact makes our estimator computationally convenient.

\begin{Remark}
Although we assume that $\alpha(\tau)$ does not depend on the time period, we can relax this assumption.
Even if the QTE parameter is $\alpha_t(\tau)$, we can estimate $\alpha_t(\tau)$ in a similar manner.
However, in such a case, our second step estimation requires non-convex optimization with respect to $a_1, \cdots, a_T$.
Hence, if $T$ is large, this estimation method becomes computationally demanding.
\end{Remark}

\subsection{Identification}

In this section, we show that $\alpha(\tau)$ and $\beta(\tau) = (\beta_1(\tau)', \cdots , \beta_T(\tau)')'$ uniquely solve the limit problems.
We define 
\begin{equation}
\beta_t(a,\tau) \equiv \text{arg} \min_{b_t \in \mathcal{B}_t} E[ R_{\tau}(W_{it};a,b_t)], \label{beta limit problem}
\end{equation}
and
\begin{equation}
\alpha^*(\tau) \in \text{arg} \min_{a \in \mathcal{A}} \frac{1}{T} \sum_{t=1}^T \left\|D^t(v;a,\beta(a,\tau)) \right\|_{L_2}^2, \nonumber
\end{equation}
where $\beta(a,\tau) \equiv (\beta_1(a,\tau)', \cdots , \beta_T(a,\tau)')'$ and $D^t(v;a,b) \equiv E[g_t(\mathbf{W}_{i};a,b,v)]$.
Hence, to prove consistency, we need to show that $\alpha^*(\tau)$ is unique and $\alpha^*(\tau) = \alpha(\tau)$.

We define $e_t(a,\tau,\mathbf{z}) \equiv P ( Y_{it} \leq X_{it}'a + Z_{it}'\beta_t(a,\tau) | \mathbf{Z}_i=\mathbf{z} )$ and impose the following assumptions.

\begin{Assumption}
(i) A matrix $E[Z_{it}Z_{it}']$ has full rank for all $t$, and $E[X_{it}X_{it}']$ has full rank for some $t$.
(ii) For all $t$, $E[|Y_{it}|]$ is finite.
(iii) For all $t$, $a$, and $\tau$, $\beta_t(a,\tau)$ uniquely solves (\ref{beta limit problem}).
\end{Assumption}

\begin{Assumption}
For all $t$ and $a \in \mathcal{A}$, $e_t(a,\tau,\mathbf{z}) = \tau$ for some $\mathbf{z} \in \mathcal{Z}$.
\end{Assumption}

When the support of $(X_{i1},X_{i2})$ is $\{(0,0),(0,1)\}$, we have $E[X_{i1}^2] = 0$ but $E[X_{i2}^2]$ is positive.
Hence, Assumption 3 (i) holds in standard DID settings.
Assumption 4 is a technical condition that is satisfied in many situations.
Using the proof of Theorem 2 in \cite{angrist2006quantile}, it follows from the first-order condition of (\ref{first step estimator}) that we have $E\left[ \left( \mathbf{1}\{Y_{it} \leq X_{it}'a + Z_{it}'\beta_t(a,\tau)\} - \tau \right) Z_{it} \right]=0$, which implies that $E\left[ \left( e_t(a,\tau,\mathbf{Z}_i) - \tau \right) Z_{it} \right]=0$.
Hence, we have $E[e_t(a,\tau,\mathbf{Z}_{i})]=\tau$ because $Z_{it}$ contains a constant.
When $\mathbf{Z}_i$ has continuous covariates and $e_t(a,\tau,\mathbf{z})$ is continuous in $\mathbf{z}$, $e_t(a,\tau,\mathbf{z}) = \tau$ holds for some $\mathbf{z} \in \mathcal{Z}$.
Even when all covariates are discrete, if $Z_{it}$ is time invariant and the model is saturated, that is, the cardinality of $\mathcal{Z}$ is equal to the dimension of $\beta_t(a,\tau)$, then we have $e_t(a,\tau,\mathbf{z})= \tau$ for all $\mathbf{z} \in \mathcal{Z}$.

\begin{Theorem}
Suppose that (\ref{Model}) and Assumptions 1--4, A.1, and A.2 hold.
Then, for all $\tau \in (0,1)$, $\alpha(\tau)$ and $\beta(\tau)$ uniquely solve the limit problems.
That is, we have $\beta_t(\alpha(\tau),\tau) = \beta_t(\tau)$ and
\begin{equation}
\frac{1}{T} \sum_{t=1}^T \left\| D^t(v;a,\beta(a,\tau)) \right\|_{L_2}^2=0,  \ \ a \in \mathcal{A} \ \ \Leftrightarrow \ \  a = \alpha(\tau). \label{Identification Estimator}
\end{equation}
\end{Theorem}

Theorem 1 implies that $\alpha(\tau)$ minimizes $\frac{1}{T} \sum_{t=1}^T \left\|D^t(v;a,\beta(a,\tau)) \right\|_{L_2}^2$.
Hence, if the objective function of (\ref{Estimator alpha}) converges to $\frac{1}{T} \sum_{t=1}^T \left\|D^t(v;a,\beta(a,\tau)) \right\|_{L_2}^2$ uniformly, then we obtain the consistency of $\hat{\alpha}(\tau)$.

\subsection{Asymptotic distribution}

In this section, we show the uniform asymptotic normality of our estimator and prove the validity of the nonparametric bootstrap.
Our asymptotic result also implies that our estimator is consistent.

Let $\mathcal{T}$ be a closed subset of $[\epsilon, 1- \epsilon]$ for $\epsilon > 0$.
 In addition, we define $J_t^b(a,\tau) \equiv E\left[ f_{Y_t-X_t'a|Z_t}(Z_{it}'\beta_t(a,\tau)|Z_{it}) Z_{it} Z_{it}' \right]$ and $J_t^b(\tau) \equiv J_t^b(\alpha(\tau),\tau)$.
The following assumption is sufficient for the consistency of $\hat{\alpha}(\tau)$ and $\hat{\beta}_t(\tau)$.

\begin{Assumption}
(i) The data $\{\mathbf{W}_i\}_{i=1}^n$ are independent and identically distributed.
(ii) For all $\tau \in \mathcal{T}$, $\alpha(\tau)$ and $\beta_{t}(\tau)$ are contained in the compact parameter spaces $\mathcal{A}$ and $\mathcal{B}_t$, respectively.
(iii) For all $t$, $E[|Y_{it}|]$ is finite, and $\mathcal{X}_{1:T}$ and $\mathcal{Z}$ are bounded.
(iv) For all $a$, $t$, and $\tau$, $\beta_t(a,\tau)$ uniquely solves (\ref{beta limit problem}).
(v) For all $a$ and $t$, the conditional density $f_{Y_t-X_t'a|Z_t}(y|z_t)$ exists and $f_{Y_t-X_t'a|Z_t}(y|z_t)$ is continuous in $y$ and bounded above.
(vi) For all $t$ and $\tau$, $J_t^b(a,\tau)$ has full rank for all $a \in \mathcal{A}$ and $\tau \in \mathcal{T}$, and $J_t^b(a,\tau)$ is continuous in $a$ at $\alpha(\tau)$.
(vii) For all $t$, $F_{Y_t|\mathbf{X},\mathbf{Z}}(y|\mathbf{x},\mathbf{z})$ is uniformly continuous in $y$.
(viii) For all $t$, $\beta_t(a,\tau)$ is continuously differentiable in $a$ and its derivative $B_t(a,\tau) \equiv \frac{\partial}{\partial a'} \beta_t(a,\tau)$ is bounded.
\end{Assumption}

Condition (iii) imposes that $X_{it}$ and $Z_{it}$ are bounded.
If $X_{it}$ and $Z_{it}$ are unbounded, then for $q_t(x_t,z_t,\tau)$ to be monotonically increasing in $\tau$, $\alpha(\tau) = \alpha$ and $\beta_t(\tau) = \beta_t$ must hold.
Hence, we assume the boundedness of $\mathcal{X}_{1:T}$ and $\mathcal{Z}$.
Condition (v) means that there exists a continuous density $f_{Y_t-X_t'a|Z_t}(y|z_t)$ for all $a \in \mathcal{A}$.
Because we have
\begin{eqnarray*}
F_{Y_t-X_t'a|Z_t}(y|z_t) &=& \int F_{Y_t|X_t,Z_t}(y + x_t'a |x_t, z_t) dF_{X_t}(x_t),
\end{eqnarray*}
we obtain $f_{Y_t-X_t'a|Z_t}(y|z_t) = \int f_{Y_t|X_t,Z_t}(y + x_t'a |x_t, z_t) dF_{X_t}(x_t)$ if $f_{Y_t|X_t,Z_t}(y|x_t,z_t)$ is bounded.
Hence, condition (v) holds if $f_{Y_t|X_t,Z_t}(y|x_t,z_t)$ is bounded and continuous in $y$.
In addition, this implies that $J_t^b(a,\tau)$ is continuous in $a$ if $\beta_t(a,\tau)$ is continuous in $a$.

We define
\begin{eqnarray}
\gamma_1^t(v;a,\tau) &\equiv & E \left[ f_{Y_t|\mathbf{X},\mathbf{Z}}(X_{it}'a + Z_{it}'\beta_t(a,\tau)|\mathbf{X}_i,\mathbf{Z}_i) \omega(\mathbf{X}_i,\mathbf{Z}_i,v) (X_{it} + B_t(a,\tau)'Z_{it}) \right] \nonumber \\
\Gamma_1^t(v;a,\tau) &\equiv & \gamma_1^t(v;a,\tau) - \frac{1}{T} \sum_{s=1}^T \gamma_1^s(v;a,\tau) \nonumber \\
\gamma_2^{t,s}(v;a,b) &\equiv & \begin{cases}
    \frac{T-1}{T} E \left[ f_{Y_t|\mathbf{X},\mathbf{Z}}(X_{it}'a + Z_{it}'b_t|\mathbf{X}_i,\mathbf{Z}_i) \omega(\mathbf{X}_i,\mathbf{Z}_i,v) Z_{it} \right], & \text{if $s=t$} \\
    -\frac{1}{T} E \left[ f_{Y_s|\mathbf{X},\mathbf{Z}}(X_{is}'a + Z_{is}'b_s|\mathbf{X}_i,\mathbf{Z}_i) \omega(\mathbf{X}_i,\mathbf{Z}_i,v) Z_{is} \right], & \text{if $s \neq t$}
  \end{cases}, \nonumber \\
\Gamma_2^t(v;a,b) &\equiv & \left( \gamma_2^{t,1}(v;a,b)' , \cdots , \gamma_2^{t,T}(v;a,b)' \right)', \nonumber
\end{eqnarray}
$\Gamma_1^t(v;\tau) \equiv \Gamma_1^t(v;\alpha(\tau),\tau)$, and $\Gamma_2^t(v;\tau) \equiv \Gamma_2^t(v;\alpha(\tau),\beta(\tau))$.
Then, the following assumption is required to derive the asymptotic distribution of the estimator.

\begin{Assumption}
(i) For all $a$, $t$, and $\tau$, $\alpha(\tau)$ and $\beta_t(a,\tau)$ are the inner points of $\mathcal{A}$ and $\mathcal{B}_t$, respectively.
(ii) A family of functions $\{y \mapsto f_{Y_t-X_t'a|Z_t}(y|z): a \in \mathcal{A}\}$ is equicontinuous for all $z \in \mathcal{Z}$.
(iii) There exists the conditional density $f_{Y_t|\mathbf{X},\mathbf{Z}}(y|\mathbf{x},\mathbf{z})$ and $f_{Y_t|\mathbf{X},\mathbf{Z}}(y|\mathbf{x},\mathbf{z})$ is uniformly continuous in $y$ and bounded above.
(vi) For all $t$, a family of functions $\{a \mapsto B_t(a,\tau): \tau \in \mathcal{T}\}$ is equicontinuous.
(v) There exists $c>0$ such that $T^{-1} \sum_{t=1}^T\|\Gamma_1^t(v;\tau)'a\|_{L_2}^2 \geq c^2 \|a\|^2$ for all $a \in \mathbb{R}^{d_X}$ and $\tau \in \mathcal{T}$.
\end{Assumption}

To derive the asymptotic distribution of $\hat{\alpha}(\tau)$, we need to show that $\sqrt{n} (\tilde{\beta}_t(a,\tau) - \beta_t(a, \tau) )$ converges in distribution uniformly in $a \in \mathcal{A}$.
We use condition (ii) to show this result.
Similarly, we need condition (iv) to show a uniform approximation of $\hat{\alpha}(\cdot)$.
Condition (v) indicates that the rank condition holds uniformly in $\tau \in \mathcal{T}$.

\begin{Theorem}
Suppose that (\ref{Model}) and (\ref{Identification Estimator}) hold for all $\tau \in \mathcal{T}$.
Under Assumptions 5 and 6, uniformly in $\tau \in \mathcal{T}$ we obtain
\begin{eqnarray}
\sqrt{n} (\hat{\alpha}(\tau)-\alpha(\tau)) &=& - \frac{1}{\sqrt{n}} \sum_{i=1}^n \mathbb{A}(\mathbf{W}_i; \tau) + o_p(1) \label{Asymptotic Normality}
\end{eqnarray}
and
\begin{equation}
\sqrt{n}(\hat{\beta}_t(\tau)-\beta_t(\tau)) = - \frac{1}{\sqrt{n}} \sum_{i=1}^n \left\{ J_t^b(\tau)^{-1} r_{\tau}(W_{it};\alpha(\tau),\beta_t(\tau)) - \mathbb{A}(\mathbf{W}_i; \tau) \right\} + o_p(1), \label{Asymptotic Normality first step estimator}
\end{equation}
where
\begin{eqnarray}
\mathbb{A}(\mathbf{W}_i; \tau) &\equiv & \Delta_1(\tau)^{-1} \left\{ \xi(\mathbf{W}_i;\tau) - \Delta_{12}(\tau) l(\mathbf{W}_i;\tau) \right\}, \nonumber \\
\xi(\mathbf{W}_i;\tau) &\equiv & \frac{1}{T} \sum_{t=1}^T \left[ \int_{\mathcal{V}} \Gamma_1^t(v;\tau) \omega(\mathbf{X}_i,\mathbf{Z}_i,v) dv \right] \mathbf{1}\{Y_{it} \leq X_{it}'\alpha(\tau)+Z_{it}'\beta_t(\tau)\}, \nonumber \\
l(\mathbf{W}_i;\tau) &\equiv & \left( r_{\tau}(W_{i1};\alpha(\tau),\beta_t(\tau))' J_1^b(\tau)^{-1}, \cdots,  r_{\tau}(W_{iT};\alpha(\tau),\beta_t(\tau))'J_T^b(\tau)^{-1} \right)', \nonumber \\
r_{\tau}(W_{it};a,b_t) &\equiv & \left( \tau - \mathbf{1}\{Y_{it} \leq X_{it}'a + Z_{it}'b_t\} \right)Z_{it}, \nonumber \\
\Delta_1(\tau) &\equiv & \frac{1}{T} \sum_{t=1}^T \int_{\mathcal{V}} \Gamma_1^t(v;\tau)\Gamma_1^t(v;\tau)' dv, \nonumber \\
\Delta_{12}(\tau) &\equiv & \frac{1}{T} \sum_{t=1}^T \int_{\mathcal{V}} \Gamma_1^t(v;\tau)\Gamma_2^t(v;\tau)' dv. \nonumber
\end{eqnarray}
\end{Theorem}

The proof of this theorem is based on arguments similar to those in \cite{brown2002weighted}, \cite{chen2003estimation}, and \cite{torgovitsky2017minimum}.

The following corollary follows immediately from Theorem 2.

\begin{Corollary}
We define $\Sigma(\tau,\tilde{\tau}) \equiv E[\mathbb{A}(\mathbf{W}_i; \tau) \mathbb{A}(\mathbf{W}_i; \tilde{\tau})']$.
Under the assumptions of Theorem 2, $\sqrt{n} (\hat{\alpha}(\cdot) - \alpha(\cdot) )$ converges weakly to a zero mean Gaussian process $\mathbb{Z}(\cdot)$ with covariance function $\Sigma(\tau,\tilde{\tau})$.
\end{Corollary}

We consider the case in which $T=2$, $X_{it}$ is scalar, and there are no covariates.
In this case, we have
\begin{eqnarray}
\xi(\mathbf{W}_i;\tau) &=& \frac{1}{4} \left[ \int_{\mathcal{V}} \gamma_1(v;\tau) \omega(\mathbf{X}_i,v) dv \right] (\mathbf{1}\{U_{i1} \leq \tau\} - \mathbf{1}\{U_{i2} \leq \tau\} ), \nonumber \\
\Delta_{12}(\tau) l(\mathbf{W}_i;\tau) &=& \frac{1}{4} \left[ \int_{\mathcal{V}} \gamma_1(v;\tau) \gamma_2^1(v;\tau) dv \right] J_1^b(\tau)^{-1}( \tau - \mathbf{1}\{U_{i1} \leq \tau\} ) \nonumber \\
& & - \frac{1}{4} \left[ \int_{\mathcal{V}} \gamma_1(v;\tau) \gamma_2^2(v;\tau) dv \right] J_2^b(\tau)^{-1}( \tau - \mathbf{1}\{U_{i2} \leq \tau\} ), \nonumber 
\end{eqnarray}
where $Y_{it}(\tau) \equiv \alpha(\tau) X_{it} + \beta_t(\tau)$,
\begin{eqnarray}
\gamma_1(v;\tau) &\equiv & E\left[ \left\{ f_{Y_1|\mathbf{X}}(Y_{i1}(\tau)|\mathbf{X}_i) (X_{i1}+B_1(\alpha(\tau),\tau)) \right. \right. \nonumber \\
& & \left. \left. \ \ \ \ \ - f_{Y_2|\mathbf{X}}(Y_{i2}(\tau)|\mathbf{X}_i)(X_{i2}+B_2(\alpha(\tau),\tau)) \right\} \omega(\mathbf{X}_i,v) \right], \nonumber 
\end{eqnarray}
and $\gamma_2^t(v;\tau) \equiv E\left[ f_{Y_t|\mathbf{X}}(Y_{it}(\tau)|\mathbf{X}_i)\omega(\mathbf{X}_i,v) \right]$.
Because $J_t^b(\tau)$, $\gamma_2^1(v;\tau)$, and $\gamma_2^2(v;\tau)$ are positive, the variances of $\xi(\mathbf{W}_i;\tau)$ and $\Delta_{12}(\tau) l(\mathbf{W}_i;\tau)$ become small when $U_{i1}$ and $U_{i2}$ are positively correlated.
Specifically, if $U_{i1}=U_{i2}$, then $\xi(\mathbf{W}_i;\tau)$ is exactly equal to zero.

Let $\{\mathbf{W}_{i}^*\}_{i=1}^n$ denote a bootstrap sample drawn with replacement from $\{\mathbf{W}_i\}_{i=1}^n$.
That is, $\{\mathbf{W}_{i}^*\}_{i=1}^n$ are independently and
identically distributed from the empirical measure, conditional on the realizations $\{\mathbf{W}_i\}_{i=1}^n$.
We define $\hat{\alpha}^*(\tau)$ as the bootstrap counterpart to $\hat{\alpha}(\tau)$.
Then, we can obtain the following theorem.

\begin{Theorem}
Under the assumptions of Theorem 2, $\sqrt{n}(\hat{\alpha}^*(\cdot) - \hat{\alpha}(\cdot))$ converges weakly to the limit distribution of $\sqrt{n}(\hat{\alpha}(\cdot) - \alpha(\cdot))$ in probability.
\end{Theorem}

\begin{Remark}
Using Theorem 3, we can consider the following null hypothesis:
\begin{equation}
H_0: \, \alpha(\tau) = r(\tau) \ \ \text{for each $\tau \in \mathcal{T}$,}
\end{equation}
where $r(\cdot)$ is known or estimable.
Then, we can use the following test statistic:
\[
S_n \ \equiv \ n \int_{\mathcal{T}} \left\| \hat{\alpha}(\tau) - \hat{r}(\tau) \right\|^2 d \tau.
\]
In practice, we approximate this integration by using a grid $\mathcal{T}_n$ in place of $\mathcal{T}$.
If the null hypothesis is that $\exists \alpha, \, \alpha(\tau) = \alpha$ for each $\tau \in \mathcal{T}$, $r(\tau)$ can be estimated by $\hat{r}(\tau) = \hat{\alpha}(0.5)$ under the null hypothesis.
In this case, Theorem 3 implies that the critical value can be calculated using the quantile of
\[
S_n^* \ \equiv \ n \int_{\mathcal{T}} \left\| (\hat{\alpha}^*(\tau) - \hat{\alpha}(\tau)) -  (\hat{\alpha}^*(0.5) - \hat{\alpha}(0.5)) \right\|^2 d\tau.
\]
\end{Remark}

\section{Simulations}

\textbf{Simulation 1.}\,
Suppose that the potential outcomes are given by
\begin{eqnarray}
Y_{i1}(x) &=& \left(1 + 0.5 \Phi^{-1}(U_{i1}) \right) x + \Phi^{-1}(U_{i1}) + Z_{i}, \nonumber \\
Y_{i2}(x) &=& \left(1 + 0.5 \Phi^{-1}(U_{i2}) \right) x + 1.2 \Phi^{-1}(U_{i2}) + 1.2 Z_{i}, \nonumber 
\end{eqnarray}
where $Z_{i} \sim U(0,1)$ and $\Phi$ is the standard normal distribution function.
The observed outcomes are generated from $Y_{it}=Y_{it}(X_{it})$.
We assume that $X_{it} = \Phi(\tilde{X}_{it})$, $U_{it} = \Phi(A_i + \tilde{U}_{it})$, $(\tilde{X}_{i1},\tilde{X}_{i2},A_i)' \sim N(0,\Sigma_{\mathbf{X} A})$, and $\tilde{U}_{it} \sim N(0,1-\rho^2)$, where $\rho \in [0,1]$ and
\begin{eqnarray} 
\Sigma_{\mathbf{X} A} = \left(
    \begin{array}{ccc}
      1 & 0.5 & 0.5\rho \\
      0.5 & 1 & 0.5\rho \\
      0.5\rho & 0.5\rho & \rho^2
    \end{array}
  \right).  \nonumber 
\end{eqnarray}
Then, $U_{it}$ is uniformly distributed and $\rho$ represents the dependence between $U_{i1}$ and $U_{i2}$.
When $\rho = 0$, $U_{i1}$ and $U_{i2}$ are uncorrelated and when $\rho = 1$, $U_{i1}$ and $U_{i2}$ are perfectly correlated.
Here, we have $\alpha(0.25) = 0.66$, $\alpha(0.5) = 1$, and $\alpha(0.75) = 1.34$.

Table 1 contains the results of this experiment for two different choices of the sample size, $1000$ and $2000$, and three different choices of $\rho^2$, $0.1$, $0.5$, and $0.9$.
The number of replications is set at $1000$ throughout.
Table 1 shows the bias, standard deviation, and MSE of the estimates of $\alpha(\tau)$ for $\tau = 0.25, 0.5$, and $0.75$.
For all settings, the bias is quite small.
Table 1 shows that the standard deviation and MSE decrease in all experiments as the sample size increases.
As expected, when the correlation between $U_{i1}$ and $U_{i2}$ is high (i.e. $\rho^2 = 0.9$), the standard deviation decreases.

\begin{table}[H]
\begin{center}
\caption{Results of Simulation 1} 
   \begin{tabular}{c c r r r r r r} \hline
      & & \multicolumn{3}{c}{$n=1000$} & \multicolumn{3}{c}{$n=2000$} \\ \hline
      &  & $\rho^2=0.1$ & $\rho^2=0.5$ & $\rho^2=0.9$ & $\rho^2=0.1$ & $\rho^2=0.5$ & $\rho^2=0.9$ \\ \hline \hline
      & bias & -0.017 & -0.018 & -0.013 & -0.003 & -0.015 & -0.014 \\ 
$\tau=0.25$ & std & 0.232 & 0.236 & 0.178 & 0.157 & 0.151 & 0.112 \\ 
      & mse & 0.054 & 0.056 & 0.032 & 0.025 & 0.023 & 0.013 \\ \hline
      & bias & -0.006 & -0.011 & -0.017 & 0.001 & -0.008 & -0.008 \\ 
$\tau=0.50$ & std & 0.204 & 0.204 & 0.149 & 0.140 & 0.133 & 0.099 \\ 
      & mse & 0.042 & 0.042 & 0.022 & 0.020 & 0.018 & 0.010 \\ \hline
      & bias & -0.013 & -0.018 & -0.018 & -0.002 & -0.011 & -0.013 \\ 
$\tau=0.75$ & std & 0.232 & 0.230 & 0.177 & 0.156 & 0.156 & 0.118 \\ 
      & mse & 0.054 & 0.053 & 0.032 & 0.024 & 0.024 & 0.014 \\ \hline
  \end{tabular}
\end{center}
\end{table}

We also verify that the nonparametric bootstrap procedure works for $(n,\rho^2) = (2000,0.9)$.
We calculate 90\% and 95\% confidence intervals of $\alpha(\tau)$ to obtain the coverage probabilities for $\tau = 0.25, 0.5$, and $0.75$.
Table 2 shows the nominal and actual coverage probabilities are close in all settings.
\begin{table}[H]
\begin{center}
\caption{Coverage probabilities of Simulation 1} 
   \begin{tabular}{c c c c} \hline
      & $\tau = 0.25$ & $\tau = 0.50$ & $\tau = 0.75$ \\ \hline \hline
 90\% & 0.894 & 0.898 & 0.892 \\
 95\% & 0.938 & 0.946 & 0.940 \\ \hline
  \end{tabular}
\end{center}
\end{table}
%
\noindent
\textbf{Simulation 2.}\,
To compare our estimation method with that of \cite{athey2006identification}, we consider the following model.
We assume that $\mathcal{X}_{1:2} = \{(0,0),(0,1)\}$ and the potential outcomes are given by
\begin{eqnarray}
Y_{i1}(0) &=& \Phi^{-1}(U_{i1}), \nonumber \\
Y_{i2}(x) &=& \left( 1+0.5\Phi^{-1}(U_{i2}) \right)x + 0.5 \Phi^{-1}(U_{i2}). \nonumber 
\end{eqnarray}
The observed outcomes are generated from $Y_{i1} = Y_{i1}(0)$ and $Y_{i2}=Y_{i2}(X_{i2})$.
It is assumed that $X_{i2}= \mathbf{1}\{\tilde{X}_i+A_i \geq 0\}$, $U_{it} = \Phi(A_i + \tilde{U}_{it})$, $\tilde{X}_i \sim N(0,1)$, $A_i \sim N(0,\rho^2)$, and $\tilde{U}_{it} \sim  N(0,1-\rho^2)$, where $\rho \in [0,1]$.
Then, $G_i \equiv \mathbf{1}\{X_{i2}=1\}$ denotes an indicator for the treatment group.

Since $U_{it}$ satisfies the rank stationarity assumption, we have
\[
E[Y_{i2}(0)-Y_{i1}(0)|G_i = g] \ = \ - 0.5 E\left[ \Phi(U_{i2}) | G_i =g \right].
\]
The conditional distribution of $U_{i2}|G_i=0$ is different from that of $U_{i2}|G_i=1$; therefore, this model does not satisfy the parallel trend assumption employed in standard DID models and we cannot estimate the average treatment effect on the treated (ATT) using the standard DID estimation method.
By contrast, using our estimation method, we can estimate the quantile functions of the potential outcomes and obtain an estimate of the ATT.

From (\ref{Identification_AI1}) and (\ref{Identification_AI2}), we can estimate $F_{Y_2(0)|G=1}(y)$ and $F_{Y_2(1)|G=0}(y)$ by
\begin{eqnarray}
\hat{F}_{Y_2(0)|G=1}(y) &\equiv & \hat{F}_{Y_1|G=1}\left( \hat{F}_{Y_1|G=0}^{-1}\left( \hat{F}_{Y_2|G=0}(y) \right) \right) \ \text{and} \nonumber \\
\hat{F}_{Y_2(1)|G=0}(y) &\equiv & \hat{F}_{Y_1|G=0}\left( \hat{F}_{Y_1|G=1}^{-1}\left( \hat{F}_{Y_2|G=1}(y) \right) \right), \nonumber
\end{eqnarray}
where $\hat{F}_{Y_t|G=g}(\cdot)$ and $\hat{F}_{Y_t|G=g}^{-1}(\cdot)$ are the empirical distribution and quantile functions, respectively.
The marginal distributions of the potential outcomes and QTE can be obtained using these estimators and the empirical distributions of $Y_{i2}|G_i=0$ and $Y_{i2}|G_i=1$.
We refer to this estimator as AI estimator.

Table 3 presents the bias, standard deviation, and MSE of our estimator and the AI estimator for $n = 500$ and three different choices of $\rho^2$, $0.1$, $0.5$, and $0.9$.
For all settings, the results of our estimator are similar to those of the AI estimator.
Hence, when there are no covariates, our estimator is not worse than the AI estimator.

\begin{table}[H]
\begin{center}
\caption{Results of Simulation 2} 
   \begin{tabular}{c c r r r r r r} \hline
      & & \multicolumn{3}{c}{Our estimator} & \multicolumn{3}{c}{AI estimator} \\ \hline
      & & $\rho^2=0.1$ & $\rho^2=0.5$ & $\rho^2=0.9$ & $\rho^2=0.1$ & $\rho^2=0.5$ & $\rho^2=0.9$ \\ \hline \hline
      & bias & -0.011 & -0.026 & -0.015 & -0.003 & -0.012 & -0.027 \\ 
$\tau=0.25$ & std & 0.129 & 0.129 & 0.097 & 0.136 & 0.129 & 0.108 \\ 
      & mse & 0.017 & 0.017 & 0.010 & 0.018 & 0.017 & 0.012 \\ \hline
      & bias & -0.002 & -0.010 & -0.006 & 0.002 & -0.001 & -0.004 \\ 
$\tau=0.50$ & std & 0.116 & 0.100 & 0.069 & 0.115 & 0.097 & 0.065 \\ 
      & mse & 0.013 & 0.010 & 0.005 & 0.013 & 0.009 & 0.004 \\ \hline
      & bias & -0.005 & -0.017 & -0.010 & -0.003 & -0.004 & -0.006 \\ 
$\tau=0.75$ & std & 0.123 & 0.106 & 0.077 & 0.125 & 0.110 & 0.073 \\ 
      & mse & 0.015 & 0.012 & 0.006 & 0.016 & 0.012 & 0.005 \\ \hline
  \end{tabular}
\end{center}
\end{table}

\section{Empirical illustrations}

\subsection{The impact of insurance provision on household production}

In this section, we use our method to study the impact of an agricultural insurance program on household production.
We use the data employed by \cite{cai2016impact} to estimate the QTE of insurance provision on tobacco production.

This empirical analysis is based on data obtained from 12 tobacco production counties in the Jiangxi province of China.
Across these 12 counties, only tobacco farmers in the county of Guangchang were eligible to buy the tobacco insurance policy.
In 2003, the People's Insurance Company of China (PICC) designed and offered the first tobacco production insurance program to households in Guangchang.
Hence, we use this county as a treatment group.

The sample includes information on approximately 3,400 tobacco households during 2002 and 2003.
Table 4 provides summary statistics for 2002 and shows that treatment regions are quite different from control regions in terms of their observed characteristics.
For example, control regions include more educated people than treatment regions.
The proportion of high school- or college-educated people in the treatment regions is $0.025$, whereas that in the control regions is $0.257$.
Hence, controlling the observed characteristics is important for adjusting the differences between the treatment and control regions.

\begin{table}[H]
\begin{center}
\caption{Summary Statistics}
  \begin{tabular}{l c c c c} \hline
      & Treatment & Control & Diff & P-val on Diff \\ \hline
      Number of households & 1260 & 2128 & & \\
      Area of tobacco production (mu) & 5.578 & 4.874 & 0.705 & 0.000 \\
      Age & 41.119 & 41.522 & -0.403 & 0.173 \\
      Household size & 4.877 & 4.665 & 0.212 & 0.000 \\
      Education (Primary) & 0.367 & 0.323 & 0.044 & 0.009 \\
      Education (Secondary) & 0.602 & 0.338 & 0.263 & 0.000 \\
      Education (High school or College) & 0.025 & 0.257 & -0.232 & 0.000 \\ \hline
  \end{tabular}
\end{center}
\end{table}

We estimate the following linear-in-parameter model:
\begin{eqnarray}
Y_{i,2002} &=& Z_i'\beta_{2002}(U_{i,2002}), \nonumber \\
Y_{i,2003} &=& G_i \alpha(U_{i,2003}) + Z_i'\beta_{2003}(U_{i,2002}), \nonumber
\end{eqnarray}
where $Y_{it}$ is the tobacco production area (mu), $G_i$ is a treatment indicator equal to one for the treatment regions and zero for the control regions, and $Z_i$ is a vector of covariates including a constant term.
We estimate $\alpha(\tau)$ for $\tau = 0.1, ... , 0.9$.
Following \cite{cai2016impact}, we employ the age of the household head, household size, and education level indicators as control variables.

The main results from our method are presented in Figure 1.
The DID estimate is $0.239$, and the $95$ \% confidence interval is $[0.078,0.388]$.
We use the nonparametric bootstrap method to construct this confidence interval.
Figure 1 shows that the estimates of $\alpha(\tau)$ differ across $\tau$, and the QTE increases in $\tau$.
The impact of the insurance provision is nearly zero at the lower and middle quantiles and positive at the upper quantiles.
For $\tau = 0.8$ and $0.9$, the QTE is statistically significant.
In addition, we consider the null hypothesis that the QTEs are constant along $\tau$.
We conduct the test described in Remark 3 and calculate the test statistic and the critical value at the $0.05$ significance level.
These values are $571.9$ and $296.3$, respectively; hence, the null hypothesis is rejected.

\begin{figure}[h]
\centering
\includegraphics[width=15cm]{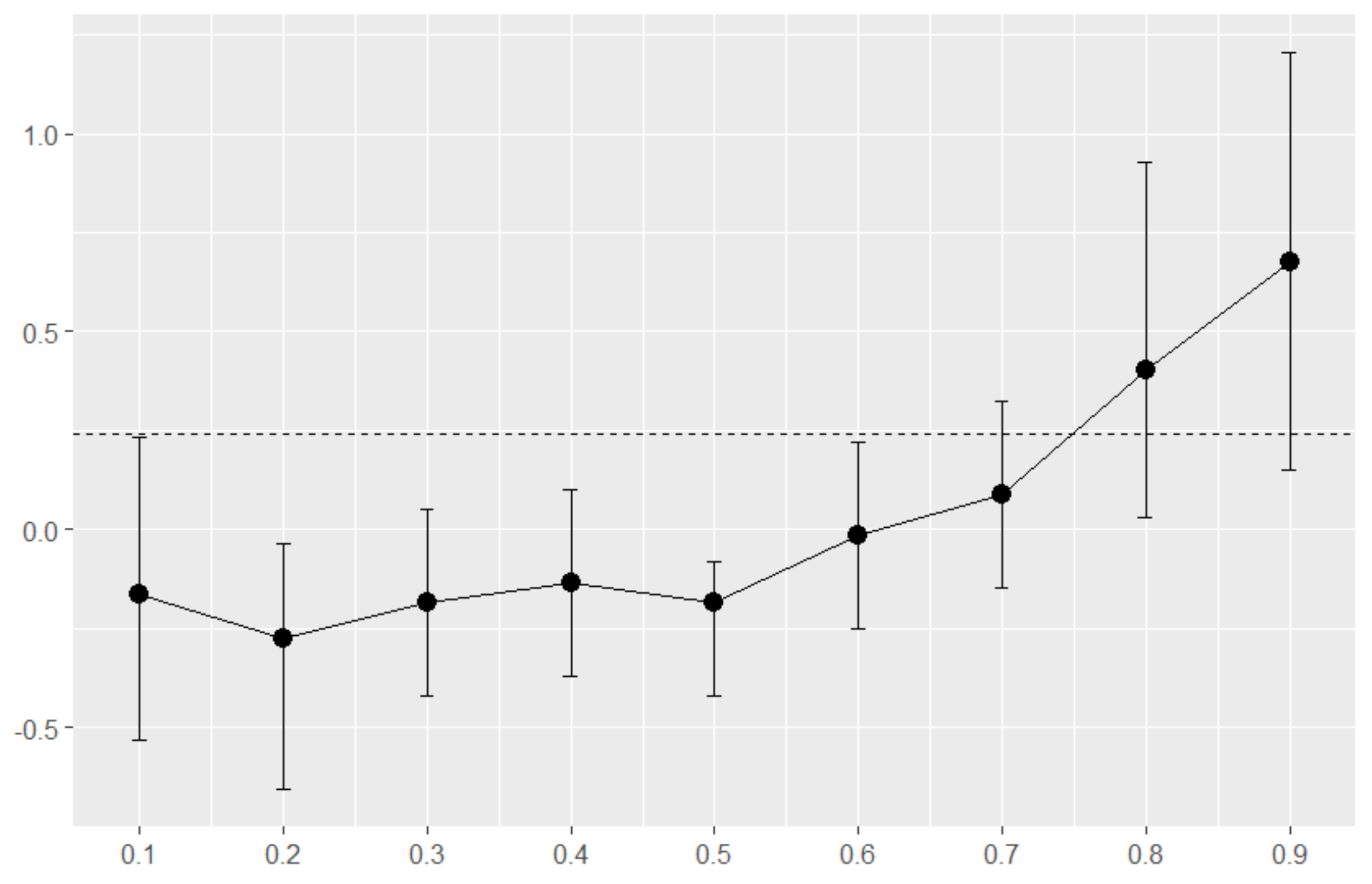}
\caption{The estimates of the QTEs and the 95 \% confidence intervals in Section 5.1. The horizontal axis measures the value of $\tau$ and the dashed line denotes the DID estimate.}
\end{figure}

\cite{cai2016impact} analyzes the welfare impact of the insurance program through the calibration.
The parameter values of the production function are chosen to match the DID (or triple difference) estimate.
From this analysis, she concludes that providing a heavily subsidized compulsory insurance program has a positive welfare impact on rural households.
However, our results show that the insurance program does not significantly change households' investment behavior at the lower and middle quantiles, and hence, may not affect household welfare at such quantiles.

\subsection{The TV effect on child cognitive development}

Next, we use our method to study the effect of TV on child cognitive development.
We use the data employed by \cite{huang2010dynamic} to estimate the QTE of TV watching on children's cognitive development.

This empirical analysis is based on a childhood longitudinal sample from NLSY79 (National Longitudinal Survey of Youth 1979).
Following \cite{huang2010dynamic}, we use a longitudinal sample of approximately 2,400 children and treat the Peabody Individual Achievement Test (PIAT) reading scores at ages 6--7 and 8--9 as $Y_{i1}$ and $Y_{i2}$, respectively.
The PIAT reading score at ages 6--7 has mean 103.0 and SD 11.7, and that at ages 8--9 has mean 104.3 and SD 14.6.
The outcome distribution at ages 8--9 is more dispersed than that at ages 6--7.
Hence, in these cases, additive time trends may not be plausible.
We use daily TV watching hours at ages 6--7 and 8--9 as the treatment variables $X_{i1}$ and $X_{i2}$, respectively and estimate the following linear-in-parameter model:
\begin{eqnarray}
Y_{it} &=& X_{it}'\alpha(U_{it}) + Z_{it}'\beta_t(U_{it}), \ \ \ \ t = 1, 2, \nonumber
\end{eqnarray}
where $Z_{it}$ is a vector of covariates including a constant term, dummy variables of race and gender, and an indicator of whether a child has 10 or more children's books at home.
In addition, we employ the Home Observation Measurement of the Environment variable (HOME), where is often used in child development research as an aggregate quality indicator of the home environment.

The main results from using our method are presented in Figure 2.
We find that the estimates of $\alpha(\tau)$ differ slightly across $\tau$ and the QTE is nearly zero at the lower and middle quantiles.
However, the impact of TV watching on the PIAT reading score is statistically significant at $\tau = 0.8$.
We consider the null hypothesis that the QTEs are zero at all quantiles.
We conduct the test described in Remark 3 and calculate the test statistic and the critical value at the $0.05$ significance level.
These values are $145.7$ and $290.5$, respectively; hence, the null hypothesis is not rejected.
Similar to \cite{huang2010dynamic}, the magnitude of the effect is quite small compared to the standard deviation of the PIAT reading score.
Therefore, the effect of TV on child cognitive development is neither statistically nor economically significant.

\begin{figure}[h]
\centering
\includegraphics[width=15cm]{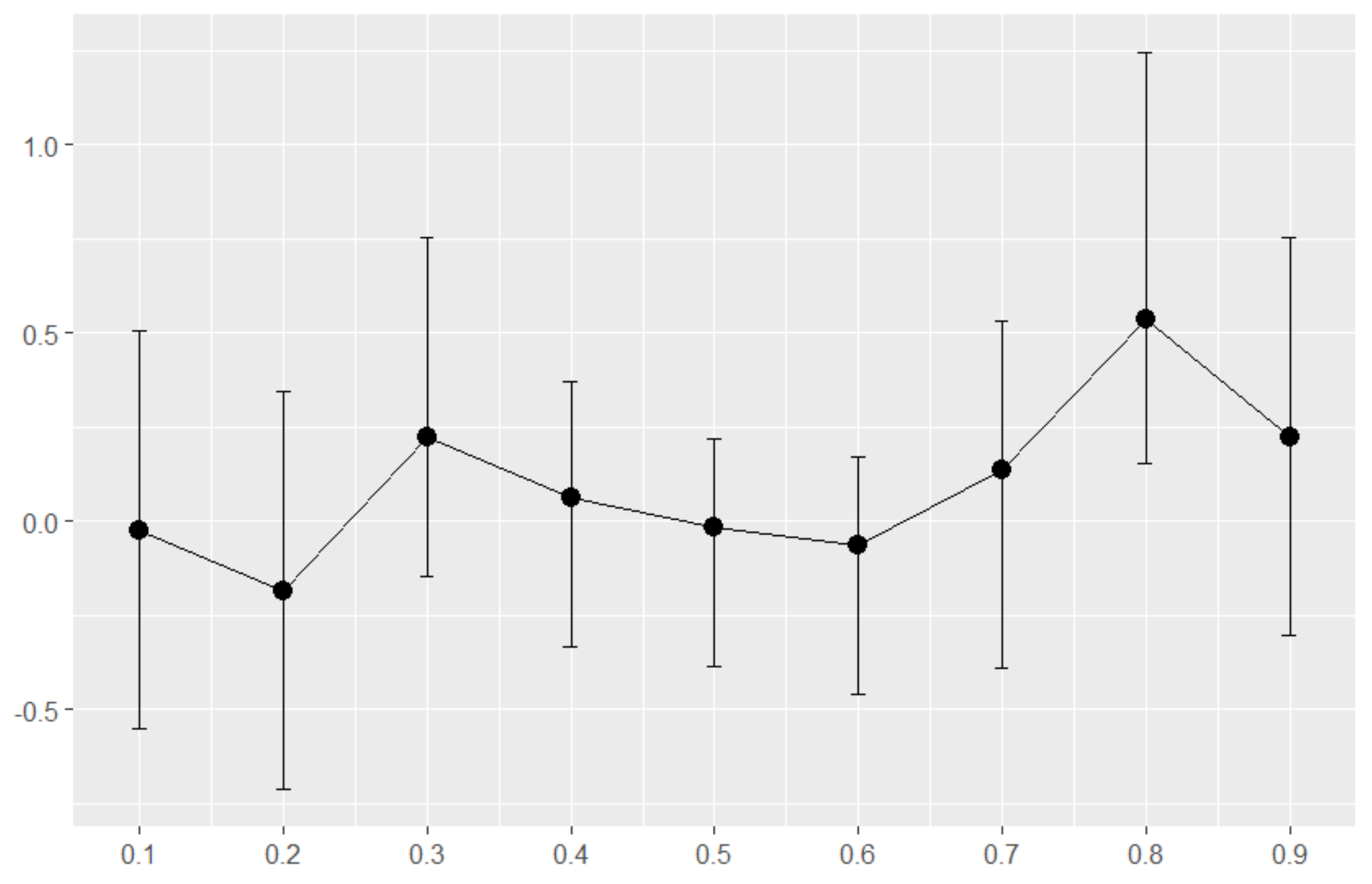}
\caption{The estimates of the QTEs and the 95 \% confidence intervals in Section 5.2. The horizontal axis represents the value of $\tau$.}
\end{figure}

\section{Conclusion}

In this study, we developed a novel estimation method for the QTE under rank invariance and rank stationarity assumptions.
Although \cite{ishihara2020identification} also explores the identification and estimation of the nonseparable panel data model under these assumptions, the minimum distance estimation using this process is computationally demanding when the dimensionality of covariates is large.
To overcome this problem, we proposed a two-step estimation method based on the quantile regression and minimum distance methods.
We then showed the uniform asymptotic properties of our estimator and the validity of the nonparametric bootstrap.
The Monte Carlo studies indicated that our estimator performs well in finite samples.
Finally, we presented two empirical illustrations to estimate the distributional effects of insurance provision on household production, and TV watching on child cognitive development.

\clearpage

\renewcommand{\theequation}{A.\arabic{equation}}
\setcounter{equation}{0}
\section*{Appendix 1: Proofs}


\begin{Assumption1}
For all $t$, $x$, and $z_t$, the conditional quantile function $q_t(x,z_t,\tau)$ is continuous and strictly increasing in $\tau$.
When $X_{it}$ is a continuous variable, then we assume that $q_t(x,z_t,\tau)$ is also continuous in $x$.
\end{Assumption1}

We define the subset $\mathcal{S}_t^m(\overline{x}) \subset \mathcal{X}_t$ in the following manner.
First, for $\bar{x} \in \mathcal{X}_1$ and $t = 1, \cdot, T$, we define $\mathcal{S}_t^0(\overline{x}) \equiv \{ x \in \mathcal{X}_t : (\bar{x},x) \in \mathcal{X}_{1,t} \}$, where $\mathcal{X}_{t,s}$ is the joint support of $(X_{it},X_{is})$.
For $m=1,2, \cdots$, we define
\begin{eqnarray}
\mathcal{S}_t^m(\overline{x}) \equiv \{x \in \mathcal{X}_t : \text{there exist $x_s \in \mathcal{S}_s^{m-1}(\overline{x})$ such that $(x,x_s) \in \mathcal{X}_{t,s}$.} \}.   \nonumber 
\end{eqnarray}
For $t = 1$, $\mathcal{S}_t^0$ becomes a singleton $\{\overline{x}\}$ and, for $t \neq 1$, $\mathcal{S}_t^0$ is the cross-section of $\mathcal{X}_{1,t}$ at $X_{i1} = \overline{x}$.

\begin{Assumption2}
(i) For all $t$, we have $\overline{\cup_{m=0}^{\infty} \mathcal{S}_t^m(\overline{x}) } = \mathcal{X}_t$ for some $\overline{x} \in \mathcal{X}_1$.
(ii) The support of $\mathbf{X}_i|\mathbf{Z}_i=\mathbf{z}$ is equal to $\mathcal{X}_{1:T}$ for all $\mathbf{z} \in \mathcal{Z}$.
\end{Assumption2}

\begin{proof}[Proof of Proposition 1]
First, we show that, if for all $x_t,x_t' \in \mathcal{X}_t$ and $\mathbf{z} = (z_1, \cdots, z_T)' \in \mathcal{Z}$, we can identify function $Q^t_{x_t',x_t|\mathbf{z}}(y)$ that satisfies
\begin{equation}
q_t(x_t',z_t,\tau) = Q^t_{x_t',x_t|\mathbf{z}}(q_t(x_t,z_t,\tau)), \label{Q_t,x,x'}
\end{equation}
then $q_t(x_t,z_t,\tau)$ is identified for all $x_t$ and $z_t$.
We define 
$$
G^t_{x_t|\mathbf{z}}(y) \equiv \int F_{Y_t|X_t,\mathbf{Z}}(Q^t_{x_t',x_t|\mathbf{z}}(y)|x_t',\mathbf{z}) dF_{X_t|\mathbf{Z}}(x_t'|\mathbf{z}).
$$
It follows from (\ref{Q_t,x,x'}) that we obtain
\begin{eqnarray}
G^t_{x_t|\mathbf{z}}(q_t(x_t,z_t,\tau)) &=& \int F_{Y_t|X_t,\mathbf{Z}}(q_t(x_t',z_t,\tau)|x_t',\mathbf{z}) dF_{X_t|\mathbf{Z}}(x_t'|\mathbf{z}) \nonumber \\
&=& \int P(U_{it} \leq \tau | X_{it} = x_t', \mathbf{Z}_i =\mathbf{z}) dF_{X_t|\mathbf{Z}}(x_t'|\mathbf{z}) \nonumber \\
&=&  P(U_{it} \leq \tau |\mathbf{Z}_i = \mathbf{z}) \ = \ \tau, \nonumber 
\end{eqnarray}
where the second equality follows from Assumption A.1.
Hence we have $q_t(x_t,z_t,\tau) = ( G_{x_t|\mathbf{z}}^t )^{-1}(\tau)$.
This implies that if we can $Q^t_{x_t',x_t|\mathbf{z}}(y)$ that satisfies (\ref{Q_t,x,x'}), then $q_t(x_t,z_t,\tau)$ is point identified.

Next, we show that for all $x_t,x_t' \in \mathcal{X}_t$ and $\mathbf{z} \in \mathcal{Z}$, we can identify function $Q^t_{x_t',x_t|\mathbf{z}}(y)$ that satisfies (\ref{Q_t,x,x'}).
For $\mathbf{x} = (x_1, \cdots , x_T)' \in \mathcal{X}_{1:T}$ and $\mathbf{z} \in \mathcal{Z}$, we have
\begin{eqnarray}
F_{Y_t|\mathbf{X},\mathbf{Z}}(q_t(x_t,z_t,\tau)|\mathbf{x},\mathbf{z}) &=& P(U_{it} \leq \tau |\mathbf{X}_i = \mathbf{x}, \mathbf{Z}_i = \mathbf{z}) \nonumber \\
&=& P(U_{is} \leq \tau |\mathbf{X}_i = \mathbf{x}, \mathbf{Z}_i = \mathbf{z}) \nonumber \\
&=& F_{Y_s|\mathbf{X},\mathbf{Z}}(q_s(x_s,z_s,\tau)|\mathbf{x},\mathbf{z}), \label{identification: main result}
\end{eqnarray}
where the second equality follows from Assumption 2.
From Assumption A.1, we obtain $q_t(x_t,z_t,\tau) = F_{Y_t|\mathbf{X},\mathbf{Z}}^{-1}\left( F_{Y_s|\mathbf{X},\mathbf{Z}}(q_s(x_s,z_s,\tau)|\mathbf{x},\mathbf{z}) | \mathbf{x},\mathbf{z} \right)$.
Hence, for any $(x_t,x_s) \in \mathcal{X}_{t,s}$, we can identify the strictly increasing function $\tilde{Q}^{s,t}_{x_s,x_t|\mathbf{z}}(y)$ such that
\begin{equation}
q_t(x_t,z_t,\tau) = \tilde{Q}^{t,s}_{x_t,x_s|\mathbf{z}} \left( q_s(x_s,z_s,\tau) \right). \label{Q_tilde s,t}
\end{equation}
This equation implies that if $q_s(x_s,z_s,\tau)$ is identified and $(x_t,x_s) \in \mathcal{X}_{t,s}$, then $q_t(x_t,z_t,\tau)$ is also identified.

Because $x_t \in \mathcal{S}_t^0(\overline{x})$ implies $(\overline{x}, x_t) \in \mathcal{X}_{1,t}$, it follows from (\ref{Q_tilde s,t}) that for any $x_t \in \mathcal{S}_t^0(\overline{x})$, we have
$$
q_t(x_t,z_t,\tau) = \tilde{Q}^{t,1}_{x_t,\overline{x}|\mathbf{z}} \left( q_1(\overline{x},z_1,\tau) \right).
$$
Next, we fix $x_t \in \mathcal{S}_t^1(\overline{x})$.
From the definition of $\mathcal{S}_t^1(\overline{x})$, there exists $x_s \in \mathcal{S}_s^0(\overline{x})$ such that $(x_t,x_s) \in \mathcal{X}_{t,s}$.
Hence, it follows from (\ref{Q_tilde s,t}) that we have
$$
q_t(x_t,z_t,\tau) = \tilde{Q}^{t,s}_{x_t,x_s|\mathbf{z}} \left( \tilde{Q}^{s,1}_{x_s,\overline{x}|\mathbf{z}} \left( q_1(\overline{x},z_1,\tau) \right) \right),
$$
which implies that we can identify $\tilde{Q}^{t,1}_{x_t,\overline{x}|\mathbf{z}}(y)$ for all $x_t \in \mathcal{S}_t^1(\overline{x})$.
By repeating this argument, for any $m \in \mathbb{N}$ and $x_t \in \mathcal{S}_t^m(\overline{x})$, we can identify $\tilde{Q}^{t,1}_{x_t,\overline{x}|\mathbf{z}}(y)$ that satisfies (\ref{Q_tilde s,t}).
Therefore, by the continuity of $q_t$, for all $x_t' \in \overline{ \cup_{m=0}^{\infty} \mathcal{S}_t^m(\overline{x}) }$, we can identify $\tilde{Q}^{t,1}_{x_t',\overline{x}|\mathbf{z}}(y)$ that satisfies (\ref{Q_tilde s,t}).
Because it follows from Assumption A.2 (i) that $\mathcal{X}_t = \overline{ \cup_{m=0}^{\infty} \mathcal{S}_t^m(\overline{x}) }$, for any $x_t',x_t \in \mathcal{X}_t$, we have
$$
q_t(x_t',z_t,\tau) = \tilde{Q}^{t,1}_{x_t',\overline{x}|\mathbf{z}}\left( \left( \tilde{Q}_{x_t,\overline{x}|\mathbf{z}}^{t,1} \right)^{-1} (q_t(x_t,z_t,\tau)) \right).
$$
Hence, we can identify function $Q^t_{x_t',x_t|\mathbf{z}}(y)$ that satisfies (\ref{Q_t,x,x'}).
\end{proof} \vspace{0.2in}

\begin{proof}[Proof of Theorem 1]
Because (\ref{moment_2}) implies $F_{Y_t - X_t'\alpha(\tau)|Z_t} \left(z_t'\beta_t(\tau)|z \right) = \tau$, it follows from the usual argument of quantile regression that we have $\beta_t(\alpha(\tau),\tau) = \beta_t(\tau)$.
Next, we show 
\begin{equation}
\frac{1}{T} \sum_{t=1}^T \|D^t(v;a,\beta(a,\tau))\|_{L_2}^2=0 \ \ \Leftrightarrow \ \  a = \alpha(\tau). \label{Thm 1 (1)}
\end{equation}
Suppose that $a = \alpha(\tau)$.
Because $\beta_t(\alpha(\tau),\tau) = \beta_t(\tau)$, it follows from (\ref{moment_second_step}) that we have
$$
D^t(v;a,\beta(a,\tau))=0.
$$
Hence, it suffices to show that the left-hand side in (\ref{Thm 1 (1)}) implies $a = \alpha(\tau)$.

Suppose that $a^* \in \mathcal{A}$ satisfies $\frac{1}{T}\sum_{t=1}^T \|D^t(v;a^*,\beta(a^*,\tau))\|_{L_2}^2=0$.
Then, for all $t$, $s$, $\mathbf{x}\in \mathcal{X}_{1:T}$, and $\mathbf{z} \in \mathcal{Z}$, we obtain
\begin{eqnarray}
 P(Y_{it} \leq X_{it}'a^* + Z_{it}'\beta_t(a^*,\tau) |\mathbf{X}_i=\mathbf{x},\mathbf{Z}_i=\mathbf{z}) \nonumber \\
= P(Y_{is} \leq X_{is}'a^* + Z_{is}'\beta_s(a^*,\tau)|\mathbf{X}_i=\mathbf{x},\mathbf{Z}_i=\mathbf{z}). \label{Thm 1 (2)}
\end{eqnarray}
Letting $\tilde{q}_t(x_t,z_t,\tau ;a^*) \equiv x_t'a^* + z_t'\beta_t(a^*,\tau)$, then (\ref{Thm 1 (2)}) implies that we have
$$
F_{Y_t|\mathbf{X},\mathbf{Z}}\left( \tilde{q}_t(x_t,z_t,\tau ;a^*) | \mathbf{x},\mathbf{z} \right) = F_{Y_s|\mathbf{X},\mathbf{Z}}\left( \tilde{q}_s(x_s,z_s,\tau ;a^*) | \mathbf{x},\mathbf{z} \right),
$$
where $\mathbf{x} = (x_1,\cdots , x_T)'$ and $\mathbf{z} = (z_1,\cdots , z_T)'$.
Similar to the proof of Proposition 1, it follows from (\ref{Thm 1 (2)}) that for all $x_t,\tilde{x}_t \in \mathcal{X}_t$, we have
\begin{equation}
\tilde{q}_t(\tilde{x}_t,z_t,\tau;a^*) = Q^t_{\tilde{x}_t,x_t|\mathbf{z}}\left( \tilde{q}_t(x_t,z_t,\tau;a^*) \right), \nonumber
\end{equation}
where $Q^t_{\tilde{x}_t,x_t|\mathbf{z}}(y)$ is defined in the proof of Proposition 1.
Then, we have
\begin{eqnarray}
G^t_{x_t|\mathbf{z}}\left( \tilde{q}_t(x_t,z_t,\tau;a^*) \right) &=& \int F_{Y_t|X_t,\mathbf{Z}}\left( \tilde{q}_t(\tilde{x}_t,z_t,\tau ;a^*)|\tilde{x},\mathbf{z} \right) dF_{X_t|Z_t}(\tilde{x}|\mathbf{z}) \nonumber \\
&=& P \left( Y_{it} \leq X_{it}'a^* + Z_{it}'\beta_t(a^*,\tau) | \mathbf{Z}_i=\mathbf{z} \right) = e_t(a^*,\tau, \mathbf{z}), \nonumber
\end{eqnarray}
where $G^t_{x_t|\mathbf{z}}(y)$ is defined in the proof of Proposition 1.
It follows from the proof of Proposition 1 that we have $\left( G^t_{x|z} \right)^{-1}(\tau) = q_t(x,z,\tau)$.
These results imply that we have
$$
\tilde{q}_t(x_t,z_t,\tau;a^*) = q_t(x_t,z_t,e_t(a^*,\tau,\mathbf{z})),
$$
that is, for all $t$, $x_t$, and $\mathbf{z}$, we have
$$
x_t'a^* + z_t'\beta_t(a^*,\tau) = x_t'\alpha(e_t(a^*,\tau,\mathbf{z})) + z_t'\beta_t(e_t(a^*,\tau,\mathbf{z})).
$$
By Assumption 3 (i), this implies that $a^* = \alpha(e_t(a^*,\tau,\mathbf{z}))$ holds for all $\mathbf{z} \in \mathcal{Z}$.
Hence, it follows from Assumption 4 that we have $a^* = \alpha(\tau)$.
\end{proof}\vspace{0.2in}

Let $\tilde{\mu}$ be a product measure $\mu \times \mu_T$, where $\mu$ and $\mu_T$ are uniform measures on $\mathcal{V}$ and $\{1, \cdots, T\}$, respectively.
Let $\|\cdot\|_{\tilde{\mu}}$ denote the $L_2$-norm with respect to $\tilde{\mu}$.
Then, we have $\|D^t(v;a,b)\|_{\tilde{\mu}}^2 = \frac{1}{T} \sum_{t=1}^T \|D^t(v;a,b)\|_{L_2}^2$ and $\|\hat{D}^t_n(v;a,b)\|_{\tilde{\mu}}^2 = \frac{1}{T} \sum_{t=1}^T \|\hat{D}^t_n(v;a,b)\|_{L_2}^2$.

\begin{proof}[Proof of Theorem 2]
For $b:\mathcal{A} \times (0,1) \mapsto \mathbb{R}^{d_Z}$, we define
\begin{eqnarray}
M^t(v;a,b,\tau) & \equiv & D^t(v;a,b(a,\tau)), \nonumber \\
M_n^t(v;a,b,\tau) & \equiv & \hat{D}^t_n(v;a,b(a,\tau)). \nonumber 
\end{eqnarray}

First, we prove $\sqrt{n}$-consistency of $\hat{\alpha}(\tau)$.
Because it follows from Lemmas 2 and 6 that we have $\sup_{\tau \in \mathcal{T}} \|\hat{\alpha}(\tau)-\alpha(\tau)\| \rightarrow_p 0$ and $\sup_{a \in \mathcal{A}, \, \tau \in \mathcal{T}} \|\tilde{\beta}(a,\tau)-\beta(a,\tau)\| \rightarrow_p 0$, we choose a positive sequence $\delta_n = o(1)$ such that $P(\|\hat{\alpha}-\alpha\|_{\infty} \geq \delta_n, \, \|\tilde{\beta}-\beta\|_{\infty} \geq \delta_n) \rightarrow 0$, where we define
\begin{eqnarray}
\|\hat{\alpha}-\alpha\|_{\infty} &\equiv & \sup_{\tau \in \mathcal{T}} \|\hat{\alpha}(\tau)-\alpha(\tau)\|, \nonumber \\
\|\tilde{\beta}-\beta\|_{\infty} &\equiv & \sup_{a \in \mathcal{A}, \, \tau \in \mathcal{T}} \|\tilde{\beta}(a,\tau)-\beta(a,\tau)\|. \nonumber
\end{eqnarray}
It follows from Lemma 5 that uniformly in $\tau$ we have
\begin{eqnarray}
\| M^t(v;\hat{\alpha}(\tau),\beta,\tau) \|_{\tilde{\mu}} + o_p(\delta_n) &\geq & \| \Gamma_1^t(v;\tau)'(\hat{\alpha}(\tau)-\alpha(\tau))\|_{\tilde{\mu}}. \nonumber 
\end{eqnarray}
By Assumption 6 (v), uniformly in $\tau$ we obtain
\begin{equation}
 \| M^t(v;\hat{\alpha}(\tau),\beta,\tau) \|_{\tilde{\mu}} \geq (c-o_p(\delta_n)) \times \|\hat{\alpha}(\tau) - \alpha(\tau)\|. \label{Thm 2 (1)}
\end{equation}
Because it follows from Lemma 3 that we have $M^t_n(v;\alpha(\tau),\beta,\tau) = O_p(n^{-1/2})$ uniformly in $v$ and $\tau$, $\| M^t(v;\hat{\alpha}(\tau),\beta,\tau) \|_{\tilde{\mu}}$ is bounded above by
\begin{eqnarray}
& & \| M^t(v;\hat{\alpha}(\tau),\beta,\tau) - M^t(v;\hat{\alpha}(\tau),\tilde{\beta},\tau) \|_{\tilde{\mu}} \nonumber \\
&+& \| M^t(v;\hat{\alpha}(\tau),\tilde{\beta},\tau) - M^t_n(v;\hat{\alpha}(\tau),\tilde{\beta},\tau) + M^t_n(v;\alpha(\tau),\beta,\tau) \|_{\tilde{\mu}} \nonumber \\
&+& \| M^t_n(v;\hat{\alpha}(\tau),\tilde{\beta},\tau) \|_{\tilde{\mu}} + O_p(n^{-1/2}), \nonumber
\end{eqnarray}
where $O_p(n^{-1/2})$ is uniform with respect to $\tau \in \mathcal{T}$.
Because a class of functions $\{r_{\tau}(w;a,b):a\in \mathcal{A}, b \in \mathcal{B}, \tau \in \mathcal{T} \}$ is Donsker, it follows from Lemmas 5 and 7 that uniformly in $\tau$ we have
\begin{eqnarray}
& & \| M^t(v;\hat{\alpha}(\tau),\beta,\tau) - M^t(v;\hat{\alpha}(\tau),\tilde{\beta},\tau) \|_{\tilde{\mu}} \nonumber \\
&\leq & \left\| M^t(v;\hat{\alpha}(\tau),\tilde{\beta},\tau) - M^t(v;\hat{\alpha}(\tau),\beta,\tau) - \Gamma_2^t(v,\hat{\alpha}(\tau),\tau)'[\tilde{\beta}(\hat{\alpha}(\tau),\tau) - \beta(\hat{\alpha}(\tau),\tau)] \right\|_{\tilde{\mu}} \nonumber \\
& & + \left\| \left( \Gamma_2^t(v,\hat{\alpha}(\tau),\tau) - \Gamma_2^t(v,\alpha(\tau),\tau) \right)'[\tilde{\beta}(\hat{\alpha}(\tau),\tau) - \beta(\hat{\alpha}(\tau),\tau)] \right\|_{\tilde{\mu}} \nonumber \\
& & + \left\| \Gamma_2^t(v,\alpha(\tau),\tau)'\left( [\tilde{\beta}(\hat{\alpha}(\tau),\tau) - \beta(\hat{\alpha}(\tau),\tau)] - [\tilde{\beta}(\alpha(\tau),\tau) - \beta(\alpha(\tau),\tau)] \right) \right\|_{\tilde{\mu}} \nonumber \\
& & + \left\| \Gamma_2^t(v,\alpha(\tau),\tau)' [\tilde{\beta}(\alpha(\tau),\tau) - \beta(\alpha(\tau),\tau)] \right\|_{\tilde{\mu}} \nonumber \\
& \leq & o_p(\delta_n) + O_p(n^{-1/2}) = O_p(n^{-1/2}). \label{Thm 2 (2)}
\end{eqnarray}
From Lemma 3, we obtain 
$$
\| M^t(v;\hat{\alpha}(\tau),\tilde{\beta},\tau) - M^t_n(v;\hat{\alpha}(\tau),\tilde{\beta},\tau) + M^t_n(v;\alpha(\tau),\beta,\tau) \|_{\tilde{\mu}} = o_p(n^{-1/2}).
$$
Hence, it follows from (\ref{Thm 2 (1)}) and (\ref{Thm 2 (2)}) that uniformly in $\tau$ we have
\begin{equation}
(c-o_p(1)) \times \|\hat{\alpha}(\tau) - \alpha(\tau)\| \leq O_p(n^{-1/2}) +  \| M^t_n(v;\hat{\alpha}(\tau),\tilde{\beta},\tau) \|_{\tilde{\mu}}. \label{Thm 2 (3)}
\end{equation}
By definition of $\hat{\alpha}(\tau)$, uniformly in $\tau$ we obtain
\begin{eqnarray}
& & \| M^t_n(v;\hat{\alpha}(\tau),\tilde{\beta},\tau) \|_{\tilde{\mu}} \nonumber \\
& \leq & \| M^t_n(v;\alpha(\tau),\tilde{\beta},\tau) \|_{\tilde{\mu}} \nonumber \\
& \leq & \left\| M^t_n(v;\alpha(\tau),\tilde{\beta},\tau) - M^t(v;\alpha(\tau),\tilde{\beta},\tau) - M^t_n(v;\alpha(\tau),\beta,\tau) \right\|_{\tilde{\mu}} \nonumber \\
& & + \left\| M^t(v;\alpha(\tau),\tilde{\beta},\tau) - \Gamma_2^t(v;\tau)'[\tilde{\beta}(\alpha(\tau),\tau) - \beta(\alpha(\tau),\tau)]  \right\|_{\tilde{\mu}} \nonumber \\
& & + \left\| \Gamma_2^t(v;\tau)'[\tilde{\beta}(\alpha(\tau),\tau) - \beta(\alpha(\tau),\tau)]  \right\|_{\tilde{\mu}} + \left\| M^t_n(v;\alpha(\tau),\beta,\tau) \right\|_{\tilde{\mu}} \nonumber \\
& \leq & o_p(n^{-1/2}) + o_p(\delta_n) + O_p(n^{-1/2}) = O_p(n^{-1/2}). \nonumber
\end{eqnarray}
As a result, from (\ref{Thm 2 (3)}), we obtain
$\| \hat{\alpha}(\tau) - \alpha(\tau) \| \leq O_p(n^{-1/2})$ uniformly in $\tau$.

Next we show (\ref{Asymptotic Normality}) by approximating $M^t_n(v;a,\tilde{\beta},\tau)$ as
\begin{eqnarray}
L^t_n(v;a,\tau) &\equiv & M_n^t(v;\alpha(\tau),\beta,\tau) + \Gamma_1^t(v;\tau)'(a-\alpha(\tau)) \nonumber \\
& & \hspace{0.7in} + \Gamma_2^t(v;\tau)' [\tilde{\beta}(\alpha(\tau),\tau) - \beta(\alpha(\tau),\tau)]. \nonumber
\end{eqnarray}
Let $\bar{\alpha}(\tau)$ be the value that provides a global minimum for $\|L^t_n(v;a,\tau)\|_{\tilde{\mu}}$.
Then, $\Gamma_1^t(v;\tau)'(\bar{\alpha}(\tau)-\alpha(\tau))$ is the $L_2(\tilde{\mu})$-projection of $-M_n^t(v;\alpha(\tau),\beta,\tau) - \Gamma_2^t(v;\tau)' [\tilde{\beta}(\alpha(\tau),\tau) - \beta(\alpha(\tau),\tau)]$ onto the subspace of $L_2(\mu)$ spanned by $\Gamma_1^t(v;\tau)$.
Hence, we obtain
\begin{eqnarray}
& & \sqrt{n} (\bar{\alpha}(\tau)-\alpha(\tau)) \nonumber \\
&=& - \Delta_1(\tau)^{-1} \sqrt{n} \int \Gamma_1^t(v;\tau) \Big\{ M_n^t(v;\alpha(\tau),\beta,\tau) \nonumber \\
& & \hspace{1.5in} + \Gamma_2^t(v;\tau)' [\tilde{\beta}(\alpha(\tau),\tau) - \beta(\alpha(\tau),\tau)] \Big\} d\tilde{\mu}(v,t), \nonumber
\end{eqnarray}
where $\Delta_1(\tau) \equiv \frac{1}{T} \sum_{t=1}^T \int_{\mathcal{V}} \Gamma_1^t(v;\tau)\Gamma_1^t(v;\tau)'  dv$ is nonsingular from Assumption 6 (v).
We observe that
\begin{eqnarray}
& &\sqrt{n} \int \Gamma_1^t(v;\tau)  M_n^t(v;\alpha(\tau),\beta,\tau) d\tilde{\mu}(v,t) \nonumber \\
&=& \frac{1}{\sqrt{n} T} \sum_{i=1}^n \sum_{t=1}^T \int_{\mathcal{V}} \Gamma_1^t(v;\tau) g_t(\mathbf{W}_i;\alpha(\tau),\beta(\tau),v)  dv. \nonumber
\end{eqnarray}
Define $Y_{it}(\tau) \equiv X_{it}'\alpha(\tau)+Z_{it}'\beta_t(\tau)$.
Because $\sum_{t=1}^T \Gamma_1^t(v;\tau) = 0$, we have
\begin{eqnarray}
& & \frac{1}{T} \sum_{t=1}^T \int_{\mathcal{V}} \Gamma_1^t(v;\tau) g_t(\mathbf{W}_i;\alpha(\tau),\beta(\tau),v) dv \nonumber \\
&=& \frac{1}{T} \sum_{t=1}^T \int_{\mathcal{V}} \Gamma_1^t(v;\tau) \omega(\mathbf{X}_i,\mathbf{Z}_i,v) dv \mathbf{1}\{Y_{it} \leq Y_{it}(\tau)\}  \nonumber \\
& & - \frac{1}{T^2} \sum_{t=1}^T \sum_{s=1}^T \int_{\mathcal{V}} \Gamma_1^t(v;\tau) \omega(\mathbf{X}_i,\mathbf{Z}_i,v) dv \mathbf{1}\{Y_{is} \leq Y_{is}(\tau)\} \nonumber \\
&=& \frac{1}{T} \sum_{t=1}^T \int_{\mathcal{V}} \Gamma_1^t(v;\tau) \omega(\mathbf{X}_i,\mathbf{Z}_i,v) dv \mathbf{1}\{Y_{it} \leq Y_{it}(\tau)\}. \nonumber
\end{eqnarray}
It follows from Lemma 7 that we have $\sqrt{n}(\tilde{\beta}(\alpha(\tau),\tau)-\beta(\alpha(\tau),\tau)) = - \frac{1}{\sqrt{n}} \sum_{i=1}^n l(\mathbf{W}_i;\tau) + o_p(1)$ uniformly in $\tau$.
Hence, uniformly in $\tau$ we obtain
\begin{eqnarray}
& & \sqrt{n} \int \Gamma_1^t(v;\tau) \Gamma_2^t(v;\tau)' [\tilde{\beta}(\alpha(\tau),\tau) - \beta(\alpha(\tau),\tau)] d\tilde{\mu}(v,t) \nonumber \\
&=& -\frac{1}{\sqrt{n}T} \sum_{i=1}^n \sum_{t=1}^T \left[ \int_{\mathcal{V}} \Gamma_1^t(v;\tau) \Gamma_2^t(v;\tau)' dv \right] l(\mathbf{W}_i;\tau) + o_p(1). \nonumber 
\end{eqnarray}
This implies that we have
\begin{eqnarray}
\sqrt{n}(\bar{\alpha}(\tau)-\alpha(\tau)) &=& -  \frac{1}{\sqrt{n}} \sum_{i=1}^n \Delta_1(\tau)^{-1} \left\{ \xi(\mathbf{W}_i;\tau) - \Delta_{12}(\tau) l(\mathbf{W}_i;\tau) \right\}  + o_p(1), \nonumber
\end{eqnarray}
where $o_p(1)$ is uniform with respect to $\tau$.
Hence, it is sufficient to show that $\|\bar{\alpha}(\tau)-\hat{\alpha}(\tau)\|=o_p(n^{-1/2})$ uniformly in $\tau$.

Because we have $\| \hat{\alpha} - \alpha \|_{\infty} = O_p(n^{-1/2})$, uniformly in $\tau$ we obtain
\begin{eqnarray}
& & \| M^t_n(v;\hat{\alpha}(\tau),\tilde{\beta},\tau) - L^t_n(v;\hat{\alpha}(\tau),\tau)  \|_{\tilde{\mu}} \nonumber \\
& \leq & \| M^t(v;\hat{\alpha}(\tau),\tilde{\beta},\tau) - M^t(v;\hat{\alpha}(\tau),\beta,\tau) - \Gamma_2^t(v;\hat{\alpha}(\tau),\tau)'[\tilde{\beta}(\hat{\alpha}(\tau),\tau) - \beta(\hat{\alpha}(\tau),\tau)] \|_{\tilde{\mu}} \nonumber \\
& & + \left\| \Gamma_2^t(v;\hat{\alpha}(\tau),\tau)'\left\{ [\tilde{\beta}(\hat{\alpha}(\tau),\tau) - \beta(\hat{\alpha}(\tau),\tau)] - [\tilde{\beta}(\alpha(\tau),\tau) - \beta(\alpha(\tau),\tau)] \right\} \right\|_{\tilde{\mu}} \nonumber \\
& & + \left\| \left\{ \Gamma_2^t(v;\hat{\alpha}(\tau),\tau) - \Gamma_2^t(v;\alpha(\tau),\tau)  \right\}'  [\tilde{\beta}(\alpha(\tau),\tau) - \beta(\alpha(\tau),\tau)] \right\|_{\tilde{\mu}} \nonumber \\
& & + \| M^t(v;\hat{\alpha}(\tau),\beta,\tau) - \Gamma_1^t(v;\tau)'(\hat{\alpha}(\tau)-\alpha(\tau)) \|_{\tilde{\mu}} \nonumber \\
& & + \| M_n^t(v;\hat{\alpha}(\tau),\tilde{\beta},\tau) - M^t(v;\hat{\alpha}(\tau),\tilde{\beta},\tau) - M_n^t(v;\alpha(\tau),\beta,\tau) \|_{\tilde{\mu}} \nonumber \\
&\leq & o_p(\|\tilde{\beta}-\beta\|_{\infty}) + o_p(\|\hat{\alpha} - \alpha \|_{\infty}) + o_p(n^{-1/2}) = o_p(n^{-1/2}). \label{Thm 2 (4)}
\end{eqnarray}
Similarly, uniformly in $\tau$ we obtain
\begin{equation}
\| M^t_n(v;\bar{\alpha}(\tau),\tilde{\beta},\tau) - L^t_n(v;\bar{\alpha}(\tau),\tau)  \|_{\tilde{\mu}} = o_p(n^{-1/2}). \label{Thm 2 (5)}
\end{equation}
Hence, it follows from (\ref{Thm 2 (4)}) and (\ref{Thm 2 (5)}) that uniformly in $\tau$ we have
\begin{eqnarray}
\| L^t_n(v;\hat{\alpha}(\tau),\tau) \|_{\tilde{\mu}} - o_p(n^{-1/2}) &\leq & \| M^t_n(v;\hat{\alpha}(\tau),\tilde{\beta},\tau)\|_{\tilde{\mu}}  \nonumber \\
& \leq & \| M^t_n(v;\bar{\alpha}(\tau),\tilde{\beta},\tau) \|_{\tilde{\mu}} \nonumber \\
& \leq & \| L^t_n(v;\bar{\alpha}(\tau),\tau)  \|_{\tilde{\mu}} + o_p(n^{-1/2}). \nonumber
\end{eqnarray}
By definition of $\bar{\alpha}(\tau)$, we have $\| L^t_n(v;\hat{\alpha}(\tau),\tau) \|_{\tilde{\mu}} \geq \| L^t_n(v;\bar{\alpha}(\tau),\tau)\|_{\tilde{\mu}}$.
Because it follows from Lemma 1 that we have $\| L^t_n(v;\bar{\alpha}(\tau),\tau)\|_{\tilde{\mu}} = O_p(n^{-1/2})$, uniformly in $\tau$ we have
\begin{eqnarray}
\| L^t_n(v;\hat{\alpha}(\tau),\tau) \|_{\tilde{\mu}}^2 = \| L^t_n(v;\bar{\alpha}(\tau),\tau)\|_{\tilde{\mu}}^2 + o_p(n^{-1}). \nonumber
\end{eqnarray}
Because $L^t_n(v;\bar{\alpha}(\tau),\tau)$ is orthogonal to $\Gamma_1^t(v;\tau)$, we obtain
\begin{eqnarray}
\| L^t_n(v;\hat{\alpha}(\tau),\tau)\|_{\tilde{\mu}}^2 &=& \| L^t_n(v;\bar{\alpha}(\tau),\tau) + \Gamma_1^t(v;\tau)'(\hat{\alpha}(\tau)-\bar{\alpha}(\tau))\|_{\tilde{\mu}}^2 \nonumber \\
&=& \| L^t_n(v;\bar{\alpha}(\tau),\tau) \|_{\tilde{\mu}}^2 + \| \Gamma_1^t(v;\tau)'(\hat{\alpha}(\tau)-\bar{\alpha}(\tau))\|_{\tilde{\mu}}^2. \nonumber
\end{eqnarray}
Hence, we have $\| \hat{\alpha}(\tau)-\bar{\alpha}(\tau)\| = o_p(n^{-1/2})$ uniformly in $\tau$ and we obtain (\ref{Asymptotic Normality}).

Finally, we show (\ref{Asymptotic Normality first step estimator}).
From the proof of Lemma 7, uniformly in $\tau$ we have
\begin{eqnarray}
o_p(1) &=& \mathbb{G}_n r_{\tau}(W_t;\alpha(\tau),\beta_t(\tau)) + o_p(1) + \sqrt{n} E r_{\tau}(W_t;\hat{\alpha}(\tau),\hat{\beta}_t(\hat{\alpha}(\tau),\tau)) \nonumber \\
&=& \mathbb{G}_n r_{\tau}(W_t;\alpha(\tau),\beta_t(\tau)) + o_p(1) + (J_t^a(\tau)+o_p(1)) \sqrt{n}(\hat{\alpha}(\tau)-\alpha(\tau)) \nonumber \\
& & + (J_t^b(\tau)+o_p(1)) \sqrt{n}(\hat{\beta}(\tau)-\beta(\tau)). \nonumber 
\end{eqnarray}
Hence, we obtain
\begin{eqnarray}
\sqrt{n}(\hat{\beta}(\tau)-\beta(\tau)) &=& - \frac{1}{\sqrt{n}} \sum_{i=1}^n \left\{ J_t^b(\tau)^{-1} r_{\tau}(W_t;\alpha(\tau),\beta_t(\tau)) - \mathbb{A}(\mathbf{W}_i ; \tau) \right\} + o_p(1), \nonumber
\end{eqnarray}
where $o_p(1)$ is uniform with respect to $\tau$.
\end{proof}\vspace{0.2in}

\begin{proof}[Proof of Corollary 1]
Similar to the proof of Lemma 1, $\{\mathbf{w} \mapsto \mathbb{A}(\mathbf{w}; \tau) : \tau \in \mathcal{T} \}$ is Donsker. Hence, from Theorem 2, $\sqrt{n} (\hat{\alpha}(\cdot) - \alpha(\cdot) )$ converges in distribution to a zero mean Gaussian process with covariance function $E[\mathbb{A}(\mathbf{W}_i; \tau) \mathbb{A}(\mathbf{W}_i; \tau')']$.
\end{proof}\vspace{0.2in}

\begin{proof}[Proof of Theorem 3]
We define
\begin{eqnarray*}
\hat{D}_{n*}^{t}(v;a,b) & \equiv & \frac{1}{n} \sum_{i=1}^n g_t\left( \mathbf{W}_{i}^*; a, b,v \right), \\
M_{n*}^t(v;a,b,\tau) &\equiv & \hat{D}_{n*}^{t}(v;a,b(a,\tau)), \\
\psi_{n*}^{t}(v;a,b,\tau) & \equiv & \sqrt{n} \left( M_{n*}^t(v;a,b,\tau) - M_{n}^t(v;a,b,\tau) \right), \\
\psi_n^{t}(v;a,b,\tau) & \equiv & \sqrt{n} \left( M_{n}^t(v;a,b,\tau) - M^t(v;a,b,\tau) \right).
\end{eqnarray*}
In addition, we define $\tilde{\beta}_t^*(a,\tau)$ and $\tilde{\beta}^*(a,\tau)$ as the bootstrap counterparts to $\tilde{\beta}_t(a,\tau)$ and $\tilde{\beta}(a,\tau)$.
From Lemma 3 and Lemma 1 in \cite{chen2003estimation}, we have for all positive sequence $\delta_n = o(1)$,
\begin{eqnarray*}
& & \sup_{ \substack{ \|\tilde{a}-a\| \leq \delta_n, \|\tilde{b}-b\| \leq \delta_n, \\ v \in \mathcal{V}} } \left| \psi_{n*}^{t}(v;\tilde{a},\tilde{b}) - \psi_{n*}^{t}(v;a,b) \right| = o_{p^*}(1) \ \ a.s. \\
& & \sup_{ \substack{ \|\tilde{a}-a\| \leq \delta_n, \|\tilde{b}-b\| \leq \delta_n, \\ v \in \mathcal{V}} } \left| \psi_{n}^{t}(v;\tilde{a},\tilde{b}) - \psi_{n}^{t}(v;a,b) \right| = o_{a.s.}(1).
\end{eqnarray*}
Here, and subsequently, superscript $\ast$ denotes a probability or moment computed under the bootstrap distribution conditional on $\{\mathbf{W}_i\}_{i=1}^n$.
From Lemmas 8 and 9, by similar arguments given in the proof of Theorem 2, we have $\|\hat{\alpha}^*(\tau) - \hat{\alpha}(\tau)\| = O_{p^*}(n^{-1/2})$ uniformly in $\tau$.

Next we approximate $M_{n*}^t(v;a,\tilde{\beta}^*,\tau) - M_{n}^t(v;\hat{\alpha}(\tau),\tilde{\beta},\tau)$ with error $o_{p^*}(n^{-1/2})$ by the linear function $\mathcal{L}_{n*}^t(v;a,\tau)$ for $a$ in a root-$n$ neighborhood of $\hat{\alpha}(\tau)$, where
\begin{eqnarray*}
\mathcal{L}_{n*}^t(v;a,\tau) &\equiv & M_{n*}^t(v;\alpha(\tau),\beta,\tau) - M_{n}^t(v;\alpha(\tau),\beta,\tau) \\ 
& & \hspace{0.1in} + \Gamma_1^t(v; \tau)'(a-\hat{\alpha}(\tau)) + \Gamma_2^t(v;\tau)' \left[ \tilde{\beta}^*(\alpha(\tau),\tau) - \tilde{\beta}(\alpha(\tau),\tau) \right].
\end{eqnarray*}
Then, we observe that
\begin{eqnarray*}
& & \left\| M_{n*}^t(v;\hat{\alpha}^*(\tau),\tilde{\beta}^*,\tau) - M_{n}^t(v;\hat{\alpha}(\tau),\tilde{\beta},\tau) - \mathcal{L}_{n*}^t(v;\hat{\alpha}^*(\tau),\tau) \right\|_{\tilde{\mu}} \\
& \leq & \left\| n^{-1/2} \left\{ \psi_{n*}^t \left(v;\hat{\alpha}^*(\tau),\tilde{\beta}^*,\tau \right) - \psi_{n*}^t \left(v;\alpha(\tau),\beta,\tau \right) \right\} \right\|_{\tilde{\mu}} \\
& & + \left\| n^{-1/2} \left\{ \psi_{n}^t \left(v;\hat{\alpha}^*(\tau),\tilde{\beta}^*,\tau \right) - \psi_{n}^t \left(v;\hat{\alpha}(\tau), \tilde{\beta},\tau \right) \right\} \right\|_{\tilde{\mu}} \\
& & + \left\| M^t\left(v;\hat{\alpha}^*(\tau),\tilde{\beta}^*,\tau \right) - M^t\left(v;\hat{\alpha}(\tau), \tilde{\beta},\tau \right) \right. \\
& & \hspace{0.5in} \left. - \Gamma_1^t(v; \tau)'(\hat{\alpha}^*(\tau)-\hat{\alpha}(\tau)) - \Gamma_2^t(v;\tau)' \left[ \tilde{\beta}^*(\alpha(\tau),\tau) - \tilde{\beta}(\alpha(\tau),\tau) \right] \right\|_{\tilde{\mu}} \\
& \leq & \left\| M^t\left(v;\hat{\alpha}^*(\tau),\tilde{\beta}^*,\tau \right) - M^t\left(v;\hat{\alpha}(\tau),\tilde{\beta}^*,\tau \right) - \Gamma_1^t(v; \tau)'(\hat{\alpha}^*(\tau)-\hat{\alpha}(\tau)) \right\|_{\tilde{\mu}} \\
& & \hspace{0.5in} + \left\| M^t\left(v;\hat{\alpha}(\tau),\tilde{\beta}^*,\tau \right) - M^t\left(v;\hat{\alpha}(\tau),\tilde{\beta},\tau \right) \right. \\
& & \hspace{1.5in} \left. - \Gamma_2^t(v;\tau)' \left[ \tilde{\beta}^*(\alpha(\tau),\tau) - \tilde{\beta}(\alpha(\tau),\tau) \right] \right\|_{\tilde{\mu}} + o_{p^*}(n^{-1/2}).
\end{eqnarray*}
Because $\{w_t \mapsto r_{\tau}(w_t;a,b_t) : a \in \mathcal{A}, b_t \in \mathcal{B}_t \}$ is a Donsker class, it follows from Lemmas 5 and 9 that we have
\begin{eqnarray*}
& & \left\| M^t\left(v;\hat{\alpha}(\tau),\tilde{\beta}^*,\tau \right) - M^t\left(v;\hat{\alpha}(\tau),\tilde{\beta},\tau \right) - \Gamma_2^t(v;\tau)' \left[ \tilde{\beta}^*(\alpha(\tau),\tau) - \tilde{\beta}(\alpha(\tau),\tau) \right] \right\|_{\tilde{\mu}} \\
& \leq & \left\| \Gamma_2^t(v;\tau)' \left\{ \left[ \tilde{\beta}^*(\alpha(\tau),\tau) - \tilde{\beta}(\alpha(\tau),\tau) \right] - \left[ \tilde{\beta}^*(\hat{\alpha}(\tau),\tau) - \tilde{\beta}(\hat{\alpha}(\tau),\tau) \right] \right\} \right\|_{\tilde{\mu}} \\
& & \hspace{0.3in} + \left\| \left( \Gamma_2^t(v;\hat{\alpha}(\tau),\tau) - \Gamma_2^t(v;\tau) \right)' \left[ \tilde{\beta}^*(\hat{\alpha}(\tau),\tau) - \tilde{\beta}(\hat{\alpha}(\tau),\tau) \right] \right\|_{\tilde{\mu}} \\
& & \hspace{0.3in} + \left\| M^t\left(v;\hat{\alpha}(\tau),\tilde{\beta}^*,\tau \right) - M^t\left(v;\hat{\alpha}(\tau),\tilde{\beta},\tau \right) \right. \\
& & \hspace{1.5in} \left. - \Gamma_2^t(v;\hat{\alpha}(\tau),\tau)' \left[ \tilde{\beta}^*(\hat{\alpha}(\tau),\tau) - \tilde{\beta}(\hat{\alpha}(\tau),\tau) \right] \right\|_{\tilde{\mu}} \ = \ o_{p^*}\left( n^{-1/2} \right). 
\end{eqnarray*}
Similarly, we obtain
\begin{eqnarray*}
& & \left\| M^t\left(v;\hat{\alpha}^*(\tau),\tilde{\beta}^*,\tau \right) - M^t\left(v;\hat{\alpha}(\tau),\tilde{\beta}^*,\tau \right) - \Gamma_1^t(v; \tau)'(\hat{\alpha}^*(\tau)-\hat{\alpha}(\tau)) \right\|_{\tilde{\mu}} \\
& \leq & \left\| M^t\left(v;\hat{\alpha}^*(\tau),\beta,\tau \right) - M^t\left(v;\hat{\alpha}(\tau),\beta ,\tau \right) - \Gamma_1^t(v; \tau)'(\hat{\alpha}^*(\tau)-\hat{\alpha}(\tau)) \right\|_{\tilde{\mu}} \\
& & \hspace{0.3in} + \left\| M^t\left(v;\hat{\alpha}^*(\tau),\tilde{\beta}^*,\tau \right) - M^t\left(v;\hat{\alpha}^*(\tau),\beta,\tau \right) \right. \\
& & \hspace{1.5in} \left. - \Gamma_2^t(v;\tau)' \left[ \tilde{\beta}^*(\alpha(\tau),\tau) - \beta(\alpha(\tau),\tau) \right] \right\|_{\tilde{\mu}} \\
& & \hspace{0.3in} + \left\| M^t\left(v;\hat{\alpha}(\tau),\tilde{\beta}^*,\tau \right) - M^t\left(v;\hat{\alpha}(\tau),\beta,\tau \right) \right. \\
& & \hspace{1.5in} \left. - \Gamma_2^t(v;\tau)' \left[ \tilde{\beta}^*(\alpha(\tau),\tau) - \beta(\alpha(\tau),\tau) \right] \right\|_{\tilde{\mu}} \ = \ o_{p^*}(n^{-1/2}).
\end{eqnarray*}
Hence, we obtain 
\[
\left\| M_{n*}^t(v;\hat{\alpha}^*(\tau),\tilde{\beta}^*,\tau) - M_{n}^t(v;\hat{\alpha}(\tau),\tilde{\beta},\tau) - \mathcal{L}_{n*}^t(v;\hat{\alpha}^*(\tau),\tau) \right\|_{\tilde{\mu}} = o_{p^*}(n^{-1/2}).
\]
Similarly, we obtain
\[
\left\| M_{n*}^t(v;\bar{\alpha}^*(\tau),\tilde{\beta}^*,\tau) - M_{n}^t(v;\hat{\alpha}(\tau),\tilde{\beta},\tau) - \mathcal{L}_{n*}^t(v;\bar{\alpha}^*(\tau),\tau) \right\|_{\tilde{\mu}} = o_{p^*}(n^{-1/2}),
\]
where $\bar{\alpha}^*(\tau)$ is the minimizer of $\mathcal{L}_{n*}^t(v;a,\tau)$.
Therefore, it follows from Lemma 9 that we have
\[
\sqrt{n} \left( \hat{\alpha}^*(\tau) - \hat{\alpha}(\tau) \right) \ = \ - \frac{1}{\sqrt{n}} \sum_{i=1}^n \mathbb{A}(\mathbf{W}_i^*;\tau) + o_{p^*}(1).
\]
By the bootstrap theorem for the mean, the term $n^{-1/2} \sum_{i=1}^n \mathbb{A}(\mathbf{W}_i^*;\tau)$ has the same distribution as $n^{-1/2} \sum_{i=1}^n \mathbb{A}(\mathbf{W}_i;\tau)$.
This concludes the proof.
\end{proof}

\section*{Appendix 2: Lemmas}

\begin{Lemma}
Under the assumptions of Theorem 2, we have
\begin{equation}
\sup_{a \in \mathcal{A}, \, b_t \in \mathcal{B}_t, \, \tau \in \mathcal{T}} \left| \frac{1}{n} \sum_{i=1}^n R_{\tau}(W_{it};a,b_t) - E[ R_{\tau}(W_{it};a,b_t) ] \right| = o_{a.s.}(1), \label{G-C R_t}
\end{equation}
\begin{equation}
\sup_{\substack{
a \in \mathcal{A}, \, b \in \mathcal{B}, \,  v \in \mathcal{V}} } \left| \hat{D}^t_n(v;a,b) - D^t(v;a,b) \right| = o_{a.s.}(1). \label{G-C D_t}
\end{equation}
\end{Lemma}

\begin{proof}
Because two collections $\left\{ (y_t,x_t,z_t) \mapsto \mathbf{1}\{y_t -x_t'a - z_t' b_t \leq 0 \} : a \in \mathcal{A}, b_t \in \mathcal{B}_t \right\}$ and $\left\{ (y_t,x_t,z_t) \mapsto y_t -x_t'a - z_t' b_t  : a \in \mathcal{A}, b_t \in \mathcal{B}_t \right\}$ are VC-classes.
Hence, from Lemma 2.6.18 of \cite{van1996weak}, $\{R_{\tau}(\cdot;a,b_t): a \in \mathcal{A},b_t \in \mathcal{B}_t, \tau \in \mathcal{T}\}$ is also a VC-class.
This implies (\ref{G-C R_t}).

Because $\{(\mathbf{x},\mathbf{z})\mapsto v_{\mathbf{x}}'\mathbf{x} + v_{\mathbf{z}}' \mathbf{z} : (v_{\mathbf{x}}',v_{\mathbf{z}}')' \in \mathcal{V} \}$ is a VC-class, $\{(\mathbf{x},\mathbf{z})\mapsto \omega(\mathbf{x},\mathbf{z},v): v \in \mathcal{V} \}$ is also a VC-class from Lemma 2.6.18 of \cite{van1996weak}.
In addition, $\left\{ (y_t,x_t,z_t) \mapsto \mathbf{1}\{y_t -x_t'a - z_t' b_t \leq 0 \} : a \in \mathcal{A}, b_t \in \mathcal{B}_t \right\}$ is a VC-class.
Hence, $\{\mathbf{w} \mapsto g_t(\mathbf{w};a,b,v) : a \in \mathcal{A}, b \in \mathcal{B}, v \in \mathcal{V}\}$ is Donsker from Lemma 2.6.18 of \cite{van1996weak}.
This implies (\ref{G-C D_t}).
\end{proof}\vspace{0.15in}

\begin{Lemma}
Under the assumptions of Theorem 2, we have
\begin{equation}
\sup_{a \in \mathcal{A}, \, \tau \in \mathcal{T}} \left\| \tilde{\beta}_t(a,\tau) - \beta_t(a,\tau) \right\| = o_p(1). \label{unform LLN beta}
\end{equation}
\end{Lemma}

\begin{proof}
Lemma 1 implies that uniformly in $a$ and $\tau$, 
\begin{eqnarray}
E[R_{\tau}(W_{it};a,\tilde{\beta}_t(a,\tau))] &=& \frac{1}{n} \sum_{i=1}^n R_{\tau}(W_{it};a,\tilde{\beta}_t(a,\tau)) + o_{a.s.}(1) \nonumber \\
& \leq & \frac{1}{n} \sum_{i=1}^n R_{\tau}(W_{it};a,\beta_t(a,\tau)) + o_{a.s.}(1) \nonumber \\
&=& E[R_{\tau}(W_{it};a,\beta_t(a,\tau))] +o_{a.s.}(1). \nonumber
\end{eqnarray}
Pick any $\delta>0$.
Let $\{B_{\delta}^{t}(a,\tau):a \in \mathcal{A}, \tau \in \mathcal{T} \}$ be a collection of balls with diameter $\delta > 0$, each centered at $\beta_t(a,\tau)$.
Because $\rho_{\tau}(u)-\rho_{\tau}(u') \leq |u-u'|$, we have $E[R_{\tau}(W_{it};a,b_t)] - E[R_{\tau}(W_{it};\tilde{a},\tilde{b}_t)] \leq C \|(a',b_t')'-(\tilde{a}',\tilde{b}_t')'\|$.
Hence, the function $b_t \mapsto E[R_{\tau}(W_{it};a,b_t)]$ is continuous uniformly over $a \in \mathcal{A}$.
Because $\frac{\partial^2}{\partial b_t \partial b_t'}E[R_{\tau}(W_t;a,b_t)]|_{b_t=\beta_t(a,\tau)} = J_t(a,\tau)$, it follows from Assumption 5 (iv) that 
$$
\inf_{a \in \mathcal{A}, \, \tau \in \mathcal{T}} \left[ \inf_{b_t \in \mathcal{B}_t \setminus B_{\delta}^t(a,\tau)} E[R_{\tau}(W_{it};a,b_t)] - E[R_{\tau}(W_{it};a,\beta_t(a,\tau))] \right] > 0.
$$
Uniformly in $a \in \mathcal{A}$ and $\tau \in \mathcal{T}$, wp $\rightarrow$ 1 we have
\begin{eqnarray}
E[R_{\tau}(W_{it};a,\tilde{\beta}_t(a,\tau))] &<& \inf_{b_t \in \mathcal{B}_t \setminus B_{\delta}^t(a,\tau)} E[R_{\tau}(W_{it};a_t,b_t)]. \nonumber
\end{eqnarray}
Therefore, wp $\rightarrow$ 1 we have $\sup_{a \in\mathcal{A}, \tau \in \mathcal{T}}\| \tilde{\beta}_t(a,\tau) - \beta_t(a,\tau) \| \leq \delta$.
\end{proof}\vspace{0.15in}

\begin{Lemma}
Define $f(W_t;a,b_t,\tau) \equiv  \left( \tau - \mathbf{1}\{Y_t \leq X_t'a + Z_t'b_t\} \right)Z_t$.
Under the assumptions of Theorem 2, for any sequence of positive numbers $\{\delta_n\}$ that converges to zero, we have
\begin{equation}
\sup_{ \|\tilde{a}-a\| \leq \delta_n, \, \|\tilde{b}_t-b_t\| \leq \delta_n, \\ \tau \in \mathcal{T}}  \left| \mathbb{G}_n f^k_{\tilde{a},\tilde{b}_t, \tau} - \mathbb{G}_n f^k_{a,b_t, \tau} \right| = o_p(1), \label{SEC_f}
\end{equation}
\begin{equation}
\sup_{ \substack{ \|\tilde{a}-a\| \leq \delta_n, \|\tilde{b}-b\| \leq \delta_n, \\ v \in \mathcal{V}} }  \left| \sqrt{n}(\hat{D}_n^t(v;\tilde{a},\tilde{b}) -  D^t(v;\tilde{a},\tilde{b})) - \sqrt{n}(\hat{D}_n^t(v;a,b) -  D^t(v;a,b)) \right| = o_p(1), \label{SEC_D}
\end{equation}
where $f^k_{a,b_t,\tau}(w)$ is the $k$-th element of $f(w;a,b_t,\tau)$.
\end{Lemma}

\begin{proof}
From the proof of Lemma 1, $\{w \mapsto f^k_{a,b_t,\tau}(w):a \in \mathcal{A}, b_t \in \mathcal{B}_t , \tau \in \mathcal{T} \}$ and $\{\mathbf{w} \mapsto g_t(\mathbf{w};a,b,v) : a \in \mathcal{A}, b \in \mathcal{B}, v \in \mathcal{V} \}$ are Donsker.
Hence, for any $\delta_n \downarrow 0$, we obtain
\begin{eqnarray}
& & \sup_{\mathbb{P}(f^k_{\tilde{a},\tilde{b}_t,\tau} - f^k_{a,b_t,\tau})^2 \leq \delta_n}  \left| \mathbb{G}_n f^k_{\tilde{a},\tilde{b}_t,\tau} - \mathbb{G}_n f^k_{a,b_t,\tau} \right|  = o_p(1), \nonumber \\
& & \sup_{\mathbb{P}(g_{t,\tilde{a},\tilde{b}} - g_{t,a,b})^2 \leq \delta_n}  \left| \mathbb{G}_n g_{t,\tilde{a},\tilde{b}} - \mathbb{G}_n g_{t,a,b} \right|  = o_p(1), \nonumber
\end{eqnarray}
where $g_{t,a,b}(\mathbf{w}) \equiv g_t(\mathbf{w};a,b,v)$.
Because $\mathcal{X}_t$ and $\mathcal{Z}$ are bounded, for some $C>0$, we have
\begin{eqnarray}
& & \mathbb{P}(f_{\tilde{a},\tilde{b}_t, \tau}^{k} - f_{a,b_t, \tau}^{k})^2 \nonumber \\
&\leq & C E\left[ \left( \mathbf{1}\{Y_t \leq X_t'a + Z_t'b_t\} - \mathbf{1}\{Y_t \leq X_t'\tilde{a} + Z_t'\tilde{b}_t\} \right)^2 \right] \nonumber \\
&=& C E\left[ \left| \mathbf{1}\{Y_t \leq X_t'a + Z_t'b_t\} - \mathbf{1}\{Y_t \leq X_t'\tilde{a} + Z_t'\tilde{b}_t\} \right| \right] \nonumber \\
&=& C E\left[ \mathbf{1}\{X_t'\tilde{a} + Z_t'\tilde{b}_t < Y_t \leq X_t'a + Z_t'b_t\} + \mathbf{1}\{X_t'a + Z_t'b_t < Y_t \leq X_t'\tilde{a} + Z_t'\tilde{b}_t\} \right] \nonumber \\
& \leq & C \int_{\mathcal{Z}_t} \int_{\mathcal{X}_t} \left| F_{Y_t|X_t,Z_t}(x'a+z'b_t | x,z) - F_{Y_t|X_t,Z_t}(x'\tilde{a}+z'\tilde{b}_t | x,z) \right| dF_{X_t}(x) dF_{Z_t}(z). \nonumber
\end{eqnarray}
It follows from Assumption 6 (iii) that $|F_{Y_t|X_t,Z_t}(y | x,z)-F_{Y_t|X_t,Z_t}(\tilde{y}| x,z)| \leq K |y-\tilde{y}|$ for some $K > 0$.
Hence, there exists a constant $C' > 0$ such that
$$
\mathbb{P}(f_{\tilde{a},\tilde{b}_t,\tau}^{k} - f_{a,b_t,\tau}^{k})^2 \leq C' \|(a'-\tilde{a}',b_t'-\tilde{b}_t')\|.
$$
Because $\|\tilde{a}-a\| \rightarrow 0$ and $\|\tilde{b}_t-b_t\| \rightarrow 0$ imply $\mathbb{P}(f_{\tilde{a},\tilde{b}_t,\tau}^{k} - f_{a,b_t,\tau}^{k})^2 \rightarrow 0$, we have (\ref{SEC_f}).
Similarly, because $\hat{D}_n^t(v;a,b) = \mathbb{P}_n g_{t,a,b}$ and $D^t(v;a,b) = \mathbf{P} g_{t,a,b}$, we can prove (\ref{SEC_D}).
\end{proof}\vspace{0.15in}

\begin{Lemma}
Under the assumptions of Theorem 2, $D^t(v;a,\beta(a,\tau))$ is continuously differentiable in $a$, $D^t(v,a,b)$ is continuously differentiable in $b$, and
\begin{eqnarray}
\frac{\partial}{\partial a} D^t(v;a,\beta(a,\tau)) &=& \Gamma_1^t(v;a,\tau), \nonumber \\
\frac{\partial}{\partial b_s} D^t(v;a,b) &=& \gamma_2^{t,s}(v;a,b). \nonumber 
\end{eqnarray}
\end{Lemma}

\begin{proof}
First, we show the continuous differentiability of $D^t(v;a,\beta(a,\tau))$ and $D^t(v,a,b)$.
We observe that
\begin{eqnarray}
D^t(v;a,b) &=& E\left[ \left( \mathbf{1}\{Y_t \leq X_t'a + Z_t'b_t\} - \frac{1}{T} \sum_{s=1}^T \mathbf{1}\{Y_s \leq X_s'a + Z_s'b_s\} \right) \omega(\mathbf{X},\mathbf{Z},v) \right] \nonumber \\
&=& E \Big[ \Big( F_{Y_t|\mathbf{X},\mathbf{Z}}(X_t'a + Z_t'b_t|\mathbf{X},\mathbf{Z}) \nonumber \\
& & \hspace{0.6in}  - \frac{1}{T} \sum_{s=1}^T F_{Y_s|\mathbf{X},\mathbf{Z}}(X_s'a + Z_s'b_s|\mathbf{X},\mathbf{Z}) \Big) \omega(\mathbf{X},\mathbf{Z},v) \Big]. \nonumber 
\end{eqnarray}
Because we have
\begin{eqnarray}
& & \frac{\partial}{\partial b_t} E[ F_{Y_t|\mathbf{X},\mathbf{Z}}(X_t'a + Z_t'b_t|\mathbf{X},\mathbf{Z}) \omega(\mathbf{X},\mathbf{Z},v)] \nonumber \\
&=& E[ f_{Y_t|\mathbf{X},\mathbf{Z}}(X_t'a + Z_t'b_t|\mathbf{X},\mathbf{Z}) \omega(\mathbf{X},\mathbf{Z},v) Z_t], \nonumber 
\end{eqnarray}
$D^t(v;a,b)$ is continuously differentiable in $b$ and $(\partial/\partial b_s) D^t(v;a,b) = \gamma_2^{t,s}(v;a,b)$.
Similarly, we have
\begin{eqnarray}
& & \frac{\partial}{\partial a} E[ F_{Y_t|\mathbf{X},\mathbf{Z}}(X_t'a + Z_t'\beta_t(a,\tau)|\mathbf{X},\mathbf{Z}) \omega(\mathbf{X},\mathbf{Z},v)] \nonumber \\
&=& E[ f_{Y_t|\mathbf{X},\mathbf{Z}}(X_t'a + Z_t'\beta_t(a,\tau)|\mathbf{X},\mathbf{Z}) \omega(\mathbf{X},\mathbf{Z},v) (X_t + B_t(a,\tau)'Z_t) ]. \nonumber
\end{eqnarray}
Hence, $D^t(v;a,\beta(a,\tau))$ is also continuously differentiable in $a$ and $(\partial/\partial a) D^t(v;a,\beta(a,\tau)) = \Gamma_1^t(v;a,\tau)$.
\end{proof}\vspace{0.15in}

\begin{Lemma}
Under the assumptions of Theorem 2, for any sequence of positive numbers $\{\delta_n\}$ that converges to zero, we have
\begin{equation}
\sup_{\tau \in \mathcal{T}, \, \|a-\alpha(\tau)\| \leq \delta_n} \| M^t(v;a,\beta,\tau) - \Gamma_1^t(v;\tau)'(a-\alpha(\tau)) \|_{\tilde{\mu}} = o(\delta_n), \label{Differentiability gamma 1} 
\end{equation}
and
\begin{eqnarray}
\sup_{ \substack{a \in \mathcal{A}, \, \tau \in \mathcal{T}, \\ \|b-\beta\|_{\infty} \leq \delta_n } } \left\| M^t(v;a,b,\tau) - M^t(v;a,\beta,\tau) - \Gamma_2^t(v;a,\tau)'[b(a,\tau)-\beta(a,\tau)] \right\|_{\tilde{\mu}} = o(\delta_n). \label{Differentiability gamma 2} 
\end{eqnarray}
\end{Lemma}

\begin{proof}
First, we show (\ref{Differentiability gamma 1}).
Because $M^t(v;a,\beta,\tau) = D^t(v;a,\beta(a,\tau))$ is continuously differentiable in $a$, there exists $\overline{a}_{t,v,\tau}$ between $\alpha(\tau)$ and $a$ such that
\begin{eqnarray}
M^t(v;a,\beta,\tau) - M^t(v;\alpha(\tau),\beta,\tau) = \Gamma_1^t(v;\overline{a}_{t,v,\tau},\tau)'(a-\alpha(\tau)). \nonumber
\end{eqnarray}
Because $M^t(v;\alpha(\tau),\beta,\tau) = 0$, we have
\begin{eqnarray}
& &  \| M^t(v;a,\beta,\tau) - \Gamma_1^t(v;\alpha(\tau),\tau)'(a-\alpha(\tau)) \|_{\tilde{\mu}} \nonumber \\
&=&  \left\| \left( \Gamma_1^t(v;\overline{a}_{t,v,\tau},\tau) - \Gamma_1^t(v;\alpha(\tau),\tau) \right)' (a - \alpha(\tau)) \right\|_{\tilde{\mu}} \nonumber \\
&\leq & \sup_{v \in \mathcal{V}, \, \tau \in \mathcal{T}} \|\Gamma_1^t(v;\overline{a}_{t,v,\tau},\tau) - \Gamma_1^t(v;\alpha(\tau),\tau)\| \times \|a - \alpha(\tau)\|. \nonumber
\end{eqnarray}
Then, for some $C > 0$, we have
\begin{eqnarray}
& & \|\Gamma_1^t(v;\overline{a}_{t,v,\tau},\tau) - \Gamma_1^t(v;\alpha(\tau),\tau)\| \nonumber \\
&\leq & C  \max_s E\left[\| f_{Y_s|\mathbf{X},\mathbf{Z}}(X_s'\overline{a}_{t,v,\tau} + Z_s'\beta_s(\overline{a}_{t,v,\tau},\tau)|\mathbf{X},\mathbf{Z})(X_s + B_s(\overline{a}_{t,v,\tau},\tau)'Z_s) \right. \nonumber \\
& & \left. - f_{Y_s|\mathbf{X},\mathbf{Z}}(X_s'\alpha(\tau) + Z_s'\beta_s(\alpha(\tau),\tau)|\mathbf{X},\mathbf{Z})(X_s + B_s(\alpha(\tau),\tau)'Z_s) \| \right]. \nonumber
\end{eqnarray}
Hence, it follows from Assumptions 6 (iii) and (vi) that we obtain (\ref{Differentiability gamma 1}).

Next, we show (\ref{Differentiability gamma 2}).
Because $D^t(v;a,b)$ is continuously differentiable in $b$, there exists $\overline{b}_{t,v,a,\tau}$ between $b(a,\tau)$ and $\beta(a,\tau)$ such that
\begin{equation}
D^t(v;a,b(a,\tau)) - D^t(v;a,\beta(a,\tau)) = \Gamma_2^t(v;a,\overline{b}_{t,v,a,\tau})'[b(a,\tau)-\beta(a,\tau)]. \nonumber
\end{equation}
Hence, we have
\begin{eqnarray}
& & \left| M^t(v;a,b,\tau) - M^t(v;a,\beta,\tau) - \Gamma_2(v;a,\tau)'[b(a,\tau)-\beta(a,\tau)] \right| \nonumber \\
&=& \left| \left( \Gamma_2^t(v;a,\overline{b}_{t,v,a,\tau}) - \Gamma_2^t(v;a,\beta(a,\tau)) \right)'[b(a,\tau)-\beta(a,\tau)] \right| \nonumber \\
&\leq & \sup_{v \in \mathcal{V}, \, a \in \mathcal{A}, \, \tau \in \mathcal{T}} \|\Gamma_2^t(v;a,\overline{b}_{t,v,a,\tau}) - \Gamma_2^t(v;a,\beta(a,\tau)) \| \times \| b(a,\tau) - \beta(a,\tau) \|_{\infty} . \nonumber
\end{eqnarray}
Similarly to (\ref{Differentiability gamma 1}), $\sup_{v \in \mathcal{V}, \, a \in \mathcal{A}, \, \tau \in \mathcal{T}} \| \Gamma_2^t(v;a,\overline{b}_{v,a,\tau}) - \Gamma_2^t(v;a,\beta(a,\tau)) \| = o(1)$ by the uniform continuity of $f_{Y_t|\mathbf{X},\mathbf{Z}}(y|\mathbf{x},\mathbf{z})$ in $y$.
Therefore, we obtain (\ref{Differentiability gamma 2}).
\end{proof}

\begin{Lemma}
Under the assumptions of Theorem 2, we have $\sup_{\tau \in \mathcal{T}} \| \hat{\alpha}(\tau) - \alpha(\tau) \| \rightarrow_p 0$ and $\sup_{\tau \in \mathcal{T}} \| \hat{\beta}(\tau) - \beta(\tau) \| \rightarrow_p 0$.
\end{Lemma}
\begin{proof}
To prove the uniform consistency of $\hat{\alpha}(\tau)$, we show the continuity of $\|D^t(v;a,b)\|_{\tilde{\mu}}$ in $a$ and $b$.
Then, for some $C>0$, we obtain
\begin{eqnarray}
& & |D^t(v;a,b) - D^t(v;\tilde{a},\tilde{b})| \nonumber \\
&= & \left| E\left[g_t(\mathbf{W};a,b,v) - g_t(\mathbf{W};\tilde{a},\tilde{b},v) \right] \right| \nonumber \\
&\leq & C \max_{s} E\left[ | F_{Y_s|\mathbf{X},\mathbf{Z}}(X_s'a+Z_s'b_s|\mathbf{X},\mathbf{Z}) - F_{Y_s|\mathbf{X},\mathbf{Z}}(X_s'\tilde{a}+Z_s'\tilde{b}_s|\mathbf{X},\mathbf{Z}) | \right]. \nonumber 
\end{eqnarray}
Hence, from Assumption 5 (vii), $\|D^t(v;a,b)\|_{\tilde{\mu}}$ is uniformly continuous in $a$ and $b$.

We show the uniform consistency of $\hat{\alpha}(\tau)$ and $\hat{\beta}(\tau)$.
From the definition of $\hat{\alpha}(\tau)$ and Lemma 1, we have
\begin{eqnarray}
\left\| D^t\left( v;\hat{\alpha}(\tau), \tilde{\beta}(\hat{\alpha}(\tau),\tau) \right) \right\|_{\tilde{\mu}} & = & \left\| \hat{D}^t_n \left( v;\hat{\alpha}(\tau), \tilde{\beta}(\hat{\alpha}(\tau),\tau) \right) \right\|_{\tilde{\mu}} + o_p(1) \nonumber \\
&\leq & \left\| \hat{D}^t_n \left( v;\alpha(\tau), \tilde{\beta}(\alpha(\tau),\tau) \right) \right\|_{\tilde{\mu}} + o_p(1) \nonumber \\
&=& \left\| D^t\left( v;\alpha(\tau), \tilde{\beta}(\alpha(\tau),\tau) \right) \right\|_{\tilde{\mu}} + o_p(1), \label{Lem 6 (1)}
\end{eqnarray}
where $p_p(1)$ is uniform with respect to $\tau \in \mathcal{T}$.
Because $F_{Y_t|\mathbf{X},\mathbf{Z}}(y|\mathbf{x},\mathbf{z})$ is uniform continuous in $y$, it follows from Lemma 2 that uniformly in $\tau$, we have
$$
\| D^t(v;a,\tilde{\beta}(a,\tau)) \|_{\tilde{\mu}} = \| D^t(v;a,\beta(a,\tau)) \|_{\tilde{\mu}} + o_p(1).
$$
Hence, (\ref{Lem 6 (1)}) implies that uniformly in $\tau$, we obtain
\begin{equation}
\left\| D^t\left( v;\hat{\alpha}(\tau), \beta(\hat{\alpha}(\tau),\tau) \right) \right\|_{\tilde{\mu}}  \leq \left\| D^t\left( v;\alpha(\tau), \beta(\tau) \right) \right\|_{\tilde{\mu}}  + o_p(1). \label{Lem 6 (2)}
\end{equation}
Pick any $\delta>0$.
From (\ref{Identification Estimator}), Assumption 5 (ii), and continuity of $\|D^t(v;a,\beta(a,\tau))\|_{\tilde{\mu}}$, we obtain
$$
 \inf_{a \in \mathcal{A}, \, \|a-\alpha(\tau)\| > \delta}  \|D^t(v;a,\beta(a,\tau)) \|_{\tilde{\mu}} > \| D^t(v;\alpha(\tau),\beta(\tau)) \|_{\tilde{\mu}}.
$$
By (\ref{Lem 6 (2)}), wp $\rightarrow$ 1 uniformly in $\tau$ we have 
\begin{eqnarray}
\left\| D^t\left( v;\hat{\alpha}(\tau), \beta(\hat{\alpha}(\tau),\tau) \right) \right\|_{\tilde{\mu}} < \inf_{a \in \mathcal{A}, \|a-\alpha(\tau)\| > \delta}  \|D^t(v;a,\beta(a,\tau)) \|_{\tilde{\mu}}. \nonumber
\end{eqnarray}
Hence, we obtain $\sup_{\tau \in \mathcal{T}} \| \hat{\alpha}(\tau) - \alpha(\tau) \| \rightarrow_p 0$.
It follows from Assumption 5 (viii) and Lemma 2 that $\sup_{\tau \in \mathcal{T}} \|\hat{\beta}(\tau) - \beta(\tau) \| \rightarrow_p 0$.
\end{proof}

\begin{Lemma}
Under the assumptions of Theorem 2, we have
\begin{equation}
\sqrt{n}(\tilde{\beta}_t(a,\tau)-\beta_t(a,\tau)) = -J_t^b(a,\tau)^{-1} \frac{1}{\sqrt{n}} \sum_{i=1}^n  r_{\tau}(W_{it};a,\beta_t(a,\tau)) + o_p(1), \nonumber
\end{equation}
where $o_p(1)$ is uniform over $a \in \mathcal{A}$ and $\tau \in \mathcal{T}$.
\end{Lemma}

\begin{proof}
By the computational properties of the ordinary quantile regression estimator (see \cite{koenker1978regression} and \cite{angrist2006quantile}), because $Z_t$ is bounded, we obtain
\begin{equation}
\frac{1}{\sqrt{n}} \sum_{i=1}^n r_{\tau}(W_{it};a,\tilde{\beta}_t(a,\tau)) = o(1) \nonumber
\end{equation}
uniformly over $a \in \mathcal{A}$ and $\tau \in \mathcal{T}$.
From Lemmas 2 and 3, we have
\begin{eqnarray}
o(1) &=& \sqrt{n} \mathbb{E}_n r_{\tau}(W_t;a,\tilde{\beta}_t(a,\tau)) \nonumber \\
&=& \mathbb{G}_n r_{\tau}(W_t;a,\beta_t(a,\tau)) + o_p(1) + \sqrt{n} E r_{\tau}(W_t;a,\tilde{\beta}_t(a,\tau)), \label{Lemma 7 (1)}
\end{eqnarray}
where the term $o_p(1)$ is uniform over $a \in \mathcal{A}$ and $\tau \in \mathcal{T}$.
Because $E r_{\tau}(W_t;a,\beta_t(a,\tau))=0$ by first order condition, we obtain
\begin{eqnarray}
E r_{\tau}(W_t;a,\tilde{\beta}_t(a,\tau)) &=& \left( \frac{\partial}{\partial b_t'} E r_{\tau}(W_t;a,b_t) \Big|_{b_t=\bar{b}^t_{a,\tau}} \right)(\tilde{\beta}_t(a,\tau) - \beta_t(a,\tau)) \nonumber \\
&=& E\left[ f_{Y_t-X_t'a|Z_t}(Z_t'\bar{b}^t_{a,\tau}|Z_t) Z_t Z_t' \right] (\tilde{\beta}_t(a,\tau) - \beta_t(a,\tau)), \nonumber 
\end{eqnarray}
where $\bar{b}^t_{a,\tau}$ is between $\tilde{\beta}_t(a,\tau)$ and $\beta_t(a,\tau)$.
Because $\{y \mapsto f_{Y_t-X_t'a|Z_t}(y|z): a \in \mathcal{A}\}$ is equicontinuous for all $z$, we have
$$
E\left[ f_{Y_t-X_t'a|Z_t}(Z_t'\bar{b}^t_{a,\tau}|Z_t) Z_t Z_t' \right] = J_t^b(a,\tau) + o_p(1) \ \ \text{uniformly over $a \in \mathcal{A}$ and $\tau \in \mathcal{T}$.}
$$
Therefore, it follows from (\ref{Lemma 7 (1)}) that
\begin{eqnarray}
\sqrt{n} (\tilde{\beta}_t(a,\tau) - \beta_t(a,\tau)) = -J_t^b(a,\tau)^{-1} \mathbb{G}_n r_{\tau}(W_t;a,\beta_t(a,\tau)) + o_p(1), \nonumber
\end{eqnarray}
where the term $o_p(1)$ is uniform over $a \in \mathcal{A}$ and $\tau \in \mathcal{T}$.
\end{proof}\vspace{0.15in}

\begin{Lemma}
Under the assumptions of Theorem 2, for all $\tau$ we obtain
\begin{eqnarray*}
\sup_{a \in \mathcal{A}} \| \hat{\alpha}^*(\tau) - \alpha(\tau) \| &=& o_{{a.s.}^*}(1), \\
\sup_{a \in \mathcal{A}, \tau \in \mathcal{T}} \left\| \tilde{\beta}^*_t(a,\tau) - \beta_t(a,\tau) \right\| &=& o_{{a.s.}^*}(1).
\end{eqnarray*}
\end{Lemma}

\begin{proof}
From the proof of Lemma 1, uniformly in $a$ and $\tau$ we have
\begin{eqnarray*}
E\left[ R_{\tau}(W_{it};a,\tilde{\beta}^*_t(a,\tau)) \right] &=& \frac{1}{n} \sum_{i=1}^n R_{\tau}(W_{it};a,\tilde{\beta}^*_t(a,\tau)) + o_{{a.s.}^*}(1) \\
& \leq & \frac{1}{n} \sum_{i=1}^n R_{\tau}(W_{it};a,\beta_t(a,\tau)) + o_{{a.s.}^*}(1) \\
&=& E\left[ R_{\tau}(W_{it};a,\beta_t(a,\tau)) \right] + o_{{a.s.}^*}(1).
\end{eqnarray*}
We think of $b_t(a,\tau)$ and $E\left[ R_{\tau}(W_{it};a,b_t(a,\tau)) \right]$ as functions with respect to $(a,\tau) \in \mathcal{A} \times \mathcal{T}$.
Let $\ell^{\infty}_k(\mathcal{A} \times \mathcal{T})$ denote the set of all uniformly bounded, $\mathbb{R}^k$-valued functions on $\mathcal{A} \times \mathcal{T}$.
From the proof of Lemma 2, $b_t(a,\tau) \mapsto \sup_{a \in \mathcal{A}, \tau \in \mathcal{T}} E\left[ R_{\tau}(W_{it};a,b_t(a,\tau)) \right]$ is continuous as a map from $\ell^{\infty}_{\text{dim}(Z_{t})}(\mathcal{A} \times \mathcal{T})$ to $\mathbb{R}$ and this map has a unique, well-separated minimum.
Hence, we obtain $\sup_{a \in \mathcal{A}, \tau \in \mathcal{T}} \| \tilde{\beta}^*_t(a,\tau) - \beta_t(a,\tau) \| = o_{{a.s.}^*}(1)$.

Because $\{\mathbf{w} \mapsto g_t(\mathbf{w};a,b,v) : a \in \mathcal{A}, b \in \mathcal{B}, v \in \mathcal{V}\}$ is a Donsker class, we have
\[
\sup_{ \substack{ a \in \mathcal{A}, b \in \mathcal{B} \\ v \in \mathcal{V}} } \left| \hat{D}_{n*}^t(v;a,b) - D^t(v;a,b) \right| = o_{{a.s.}^*}(1).
\]
Combined with $\sup_{a \in \mathcal{A}, \tau \in \mathcal{T}} \| \tilde{\beta}^*_t(a,\tau) - \beta_t(a,\tau) \| = o_{{a.s.}^*}(1)$, this implies that uniformly in $a$ and $\tau$ we obtain
\begin{eqnarray*}
\| \hat{D}_{n*}^t(v;a,\tilde{\beta}^*_t(a,\tau)) \|_{\tilde{\mu}} &=& \| D^t(v;a,\beta_t(a,\tau)) \|_{\tilde{\mu}} + o_{{a.s.}^*}(1).
\end{eqnarray*}
By similar arguments given in the proof of Lemma 7, we obtain $\sup_{\tau \in \mathcal{T}} \| \hat{\alpha}^*(\tau) - \alpha(\tau) \| = o_{{a.s.}^*}(1)$.
\end{proof}

\begin{Lemma}
Under the assumptions of Theorem 2, uniformly in $a \in \mathcal{A}$ and $\tau \in \mathcal{T}$ we have
\[
\sqrt{n} \left( \tilde{\beta}_t^*(a,\tau) - \tilde{\beta}_t(a,\tau) \right) = - J_t^b(a,\tau)^{-1} \mathbb{G}_n^* r_{\tau}(a,\beta_t(a,\tau)) + o_{p^*}(1),
\]
where $\mathbb{E}_n r_{\tau}(a,b_t) \equiv \frac{1}{n} \sum_{i=1}^n r_{\tau}(W_{it};a,b_t)$, $\mathbb{E}_n^* r_{\tau}(a,b_t) \equiv \frac{1}{n} \sum_{i=1}^n r_{\tau}(W_{it}^*;a,b_t)$, and $\mathbb{G}_n^* r_{\tau}(a,b_t) \equiv  \sqrt{n} ( \mathbb{E}_n^* r_{\tau}(a,b_t) - \mathbb{E}_n r_{\tau}(a,b_t) )$.
\end{Lemma}

\begin{proof}
Similar to Lemma 7, uniformly in $a$ and $\tau$, we have $\sqrt{n} \mathbb{E}_n r_{\tau}(a,\tilde{\beta}_t(a,\tau)) = o(1)$ and hence
\begin{eqnarray*}
o(1) &=& \sqrt{n} \mathbb{E}_n^* r_{\tau}(a,\tilde{\beta}_t^*(a,\tau)) \\
&=& \mathbb{G}_n^* r_{\tau}(a,\beta_t(a,\tau)) + o_{p^*}(1) + \sqrt{n} \mathbb{E}_n r_{\tau}(a,\tilde{\beta}_t^*(a,\tau)).
\end{eqnarray*}
These equalities imply that
\[
o_{p^*}(1) = \mathbb{G}_n^* r_{\tau}(a,\beta_t(a,\tau)) + \sqrt{n} \left\{ \mathbb{E}_n r_{\tau}(a,\tilde{\beta}_t^*(a,\tau)) - \mathbb{E}_n r_{\tau}(a,\tilde{\beta}_t(a,\tau)) \right\}.
\]
We define $\mathbb{E} r_{\tau}(a,b_t) \equiv E[r_{\tau}(W_{it};a,b_t)]$ and $\mathbb{G}_n r_{\tau}(a,b_t) \equiv \sqrt{n} (\mathbb{E}_n r_{\tau}(a,b_t) - \mathbb{E} r_{\tau}(a,b_t))$.
Because $\{w_t \mapsto r_{\tau}(w_t;a,b_t) : a \in \mathcal{A}, b_t \in \mathcal{B}_t, \tau \in \mathcal{T} \}$ is a Donsker class, uniformly in $a$ and $\tau$ we obtain
\begin{eqnarray*}
& & \sqrt{n} \left\{ \mathbb{E}_n r_{\tau}(a,\tilde{\beta}_t^*(a,\tau)) - \mathbb{E}_n r_{\tau}(a,\tilde{\beta}_t(a,\tau)) \right\} \\
&=& \mathbb{G}_n r_{\tau}(a,\tilde{\beta}_t^*(a,\tau)) - \mathbb{G}_n r_{\tau}(a,\tilde{\beta}_t^*(a,\tau)) + \sqrt{n} \left\{ \mathbb{E} r_{\tau}(a,\tilde{\beta}_t^*(a,\tau)) - \mathbb{E} r_{\tau}(a,\tilde{\beta}_t(a,\tau)) \right\} \\
&=& o_{{a.s.}^*}(1) + \sqrt{n} \left\{ \mathbb{E} r_{\tau}(a,\tilde{\beta}_t^*(a,\tau)) - \mathbb{E} r_{\tau}(a,\tilde{\beta}_t(a,\tau)) \right\}.
\end{eqnarray*}
By similar arguments given in the proof of Lemma 7, we obtain
\[
\sqrt{n} \left( \tilde{\beta}_t^*(a,\tau) - \tilde{\beta}_t(a,\tau) \right) = - J_t^b(a,\tau)^{-1} \mathbb{G}_n^* r_{\tau}(a,\beta_t(a,\tau)) + o_{p^*}(1),
\]
where the term $o_{p^*}(1)$ is uniform over $a \in \mathcal{A}$ and $\tau \in \mathcal{T}$.
\end{proof}

\clearpage

\clearpage

\bibliographystyle{ecta}
\bibliography{nonseparable_model}

\begin{thebibliography}{36}
\newcommand{\enquote}[1]{``#1''}
\expandafter\ifx\csname natexlab\endcsname\relax\def\natexlab#1{#1}\fi

\bibitem[\protect\citeauthoryear{Abadie}{Abadie}{2002}]{abadie2002bootstrap}
\textsc{Abadie, A.} (2002): \enquote{Bootstrap tests for distributional
  treatment effects in instrumental variable models,} \emph{Journal of the
  American statistical Association}, 97, 284--292.

\bibitem[\protect\citeauthoryear{Abadie, Angrist, and Imbens}{Abadie
  et~al.}{2002}]{abadie2002instrumental}
\textsc{Abadie, A., J.~Angrist, and G.~Imbens} (2002): \enquote{Instrumental
  variables estimates of the effect of subsidized training on the quantiles of
  trainee earnings,} \emph{Econometrica}, 70, 91--117.

\bibitem[\protect\citeauthoryear{Angrist, Chernozhukov, and
  Fern{\'a}ndez-Val}{Angrist et~al.}{2006}]{angrist2006quantile}
\textsc{Angrist, J., V.~Chernozhukov, and I.~Fern{\'a}ndez-Val} (2006):
  \enquote{Quantile regression under misspecification, with an application to
  the US wage structure,} \emph{Econometrica}, 74, 539--563.

\bibitem[\protect\citeauthoryear{Athey and Imbens}{Athey and
  Imbens}{2006}]{athey2006identification}
\textsc{Athey, S. and G.~W. Imbens} (2006): \enquote{Identification and
  inference in nonlinear difference-in-differences models,}
  \emph{Econometrica}, 74, 431--497.

\bibitem[\protect\citeauthoryear{Brown and Wegkamp}{Brown and
  Wegkamp}{2002}]{brown2002weighted}
\textsc{Brown, D.~J. and M.~H. Wegkamp} (2002): \enquote{Weighted Minimum
  Mean--Square Distance from Independence Estimation,} \emph{Econometrica}, 70,
  2035--2051.

\bibitem[\protect\citeauthoryear{Cai}{Cai}{2016}]{cai2016impact}
\textsc{Cai, J.} (2016): \enquote{The impact of insurance provision on
  household production and financial decisions,} \emph{American Economic
  Journal: Economic Policy}, 8, 44--88.

\bibitem[\protect\citeauthoryear{Callaway and Li}{Callaway and
  Li}{2019}]{callaway2019quantile}
\textsc{Callaway, B. and T.~Li} (2019): \enquote{Quantile treatment effects in
  difference in differences models with panel data,} \emph{Quantitative
  Economics}, 10, 1579--1618.

\bibitem[\protect\citeauthoryear{Chamberlain}{Chamberlain}{1982}]{chamberlain1982multivariate}
\textsc{Chamberlain, G.} (1982): \enquote{Multivariate regression models for
  panel data,} \emph{Journal of Econometrics}, 18, 5--46.

\bibitem[\protect\citeauthoryear{Chen, Linton, and Van~Keilegom}{Chen
  et~al.}{2003}]{chen2003estimation}
\textsc{Chen, X., O.~Linton, and I.~Van~Keilegom} (2003): \enquote{Estimation
  of semiparametric models when the criterion function is not smooth,}
  \emph{Econometrica}, 71, 1591--1608.

\bibitem[\protect\citeauthoryear{Chernozhukov, Fern{\'a}ndez-Val, Hahn, and
  Newey}{Chernozhukov et~al.}{2013}]{chernozhukov2013average}
\textsc{Chernozhukov, V., I.~Fern{\'a}ndez-Val, J.~Hahn, and W.~Newey} (2013):
  \enquote{Average and quantile effects in nonseparable panel models,}
  \emph{Econometrica}, 81, 535--580.

\bibitem[\protect\citeauthoryear{Chernozhukov, Fernandez-Val, Hoderlein,
  Holzmann, and Newey}{Chernozhukov
  et~al.}{2015}]{chernozhukov2015nonparametric}
\textsc{Chernozhukov, V., I.~Fernandez-Val, S.~Hoderlein, H.~Holzmann, and
  W.~Newey} (2015): \enquote{Nonparametric identification in panels using
  quantiles,} \emph{Journal of Econometrics}, 188, 378--392.

\bibitem[\protect\citeauthoryear{Chernozhukov and Hansen}{Chernozhukov and
  Hansen}{2004}]{chernozhukov2004effects}
\textsc{Chernozhukov, V. and C.~Hansen} (2004): \enquote{The effects of 401 (k)
  participation on the wealth distribution: an instrumental quantile regression
  analysis,} \emph{the Review of Economics and Statistics}, 86, 735--751.

\bibitem[\protect\citeauthoryear{Chernozhukov and Hansen}{Chernozhukov and
  Hansen}{2005}]{chernozhukov2005iv}
---\hspace{-.1pt}---\hspace{-.1pt}--- (2005): \enquote{An IV model of quantile
  treatment effects,} \emph{Econometrica}, 73, 245--261.

\bibitem[\protect\citeauthoryear{Chernozhukov and Hansen}{Chernozhukov and
  Hansen}{2006}]{chernozhukov2006instrumental}
---\hspace{-.1pt}---\hspace{-.1pt}--- (2006): \enquote{Instrumental quantile
  regression inference for structural and treatment effect models,}
  \emph{Journal of Econometrics}, 132, 491--525.

\bibitem[\protect\citeauthoryear{D'Haultf{\oe}uille and
  F{\'e}vrier}{D'Haultf{\oe}uille and F{\'e}vrier}{2015}]{d2015identification}
\textsc{D'Haultf{\oe}uille, X. and P.~F{\'e}vrier} (2015):
  \enquote{Identification of nonseparable triangular models with discrete
  instruments,} \emph{Econometrica}, 83, 1199--1210.

\bibitem[\protect\citeauthoryear{D'Haultf{\oe}uille, Hoderlein, and
  Sasaki}{D'Haultf{\oe}uille et~al.}{2013}]{d2013nonlinear}
\textsc{D'Haultf{\oe}uille, X., S.~Hoderlein, and Y.~Sasaki} (2013):
  \enquote{Nonlinear difference-in-differences in repeated cross sections with
  continuous treatments,} Tech. rep., Boston College Department of Economics.

\bibitem[\protect\citeauthoryear{Evdokimov}{Evdokimov}{2010}]{evdokimov2010identification}
\textsc{Evdokimov, K.} (2010): \enquote{Identification and estimation of a
  nonparametric panel data model with unobserved heterogeneity,}
  \emph{Department of Economics, Princeton University}.

\bibitem[\protect\citeauthoryear{Feng, Vuong, and Xu}{Feng
  et~al.}{2020}]{feng2020estimation}
\textsc{Feng, Q., Q.~Vuong, and H.~Xu} (2020): \enquote{Estimation of
  heterogeneous individual treatment effects with endogenous treatments,}
  \emph{Journal of the American Statistical Association}, 115, 231--240.

\bibitem[\protect\citeauthoryear{Firpo}{Firpo}{2007}]{firpo2007efficient}
\textsc{Firpo, S.} (2007): \enquote{Efficient semiparametric estimation of
  quantile treatment effects,} \emph{Econometrica}, 75, 259--276.

\bibitem[\protect\citeauthoryear{Fr{\"o}lich and Melly}{Fr{\"o}lich and
  Melly}{2013}]{frolich2013unconditional}
\textsc{Fr{\"o}lich, M. and B.~Melly} (2013): \enquote{Unconditional quantile
  treatment effects under endogeneity,} \emph{Journal of Business \& Economic
  Statistics}, 31, 346--357.

\bibitem[\protect\citeauthoryear{Graham and Powell}{Graham and
  Powell}{2012}]{graham2012identification}
\textsc{Graham, B.~S. and J.~L. Powell} (2012): \enquote{Identification and
  estimation of average partial effects in “irregular” correlated random
  coefficient panel data models,} \emph{Econometrica}, 80, 2105--2152.

\bibitem[\protect\citeauthoryear{Havnes and Mogstad}{Havnes and
  Mogstad}{2015}]{havnes2015universal}
\textsc{Havnes, T. and M.~Mogstad} (2015): \enquote{Is universal child care
  leveling the playing field?} \emph{Journal of Public Economics}, 127,
  100--114.

\bibitem[\protect\citeauthoryear{Hoderlein and White}{Hoderlein and
  White}{2012}]{hoderlein2012nonparametric}
\textsc{Hoderlein, S. and H.~White} (2012): \enquote{Nonparametric
  identification in nonseparable panel data models with generalized fixed
  effects,} \emph{Journal of Econometrics}, 168, 300--314.

\bibitem[\protect\citeauthoryear{Huang and Lee}{Huang and
  Lee}{2010}]{huang2010dynamic}
\textsc{Huang, F. and M.-J. Lee} (2010): \enquote{Dynamic treatment effect
  analysis of TV effects on child cognitive development,} \emph{Journal of
  Applied Econometrics}, 25, 392--419.

\bibitem[\protect\citeauthoryear{Ishihara}{Ishihara}{2020}]{ishihara2020identification}
\textsc{Ishihara, T.} (2020): \enquote{Identification and estimation of
  time-varying nonseparable panel data models without stayers,} \emph{Journal
  of Econometrics}, 215, 184--208.

\bibitem[\protect\citeauthoryear{James, Lahti, and Hoynes}{James
  et~al.}{2006}]{james2006mean}
\textsc{James, S., T.~Lahti, and H.~W. Hoynes} (2006): \enquote{What mean
  impacts miss: Distributional effects of welfare reform experiments,}
  \emph{The American Economic Review}, 96, 988--1012.

\bibitem[\protect\citeauthoryear{Koenker and Bassett}{Koenker and
  Bassett}{1978}]{koenker1978regression}
\textsc{Koenker, R. and G.~Bassett} (1978): \enquote{Regression quantiles,}
  \emph{Econometrica}, 33--50.

\bibitem[\protect\citeauthoryear{Kottelenberg and Lehrer}{Kottelenberg and
  Lehrer}{2017}]{kottelenberg2017targeted}
\textsc{Kottelenberg, M.~J. and S.~F. Lehrer} (2017): \enquote{Targeted or
  Universal Coverage? Assessing Heterogeneity in the Effects of Universal Child
  Care,} \emph{Journal of Labor Economics}, 35, 609--653.

\bibitem[\protect\citeauthoryear{Martincus and Carballo}{Martincus and
  Carballo}{2010}]{martincus2010beyond}
\textsc{Martincus, C.~V. and J.~Carballo} (2010): \enquote{Beyond the average
  effects: The distributional impacts of export promotion programs in
  developing countries,} \emph{Journal of Development Economics}, 92, 201--214.

\bibitem[\protect\citeauthoryear{Matzkin}{Matzkin}{2003}]{matzkin2003nonparametric}
\textsc{Matzkin, R.~L.} (2003): \enquote{Nonparametric estimation of
  nonadditive random functions,} \emph{Econometrica}, 71, 1339--1375.

\bibitem[\protect\citeauthoryear{Melly and Santangelo}{Melly and
  Santangelo}{2015}]{melly2015changes}
\textsc{Melly, B. and G.~Santangelo} (2015): \enquote{The changes-in-changes
  model with covariates,} \emph{Universit{\"a}t Bern, Bern}.

\bibitem[\protect\citeauthoryear{Sawada}{Sawada}{2019}]{sawada2019noncompliance}
\textsc{Sawada, M.} (2019): \enquote{Noncompliance in randomized control trials
  without exclusion restrictions,} \emph{arXiv preprint arXiv:1910.03204}.

\bibitem[\protect\citeauthoryear{Stinchcombe and White}{Stinchcombe and
  White}{1998}]{stinchcombe1998consistent}
\textsc{Stinchcombe, M.~B. and H.~White} (1998): \enquote{Consistent
  specification testing with nuisance parameters present only under the
  alternative,} \emph{Econometric theory}, 14, 295--325.

\bibitem[\protect\citeauthoryear{Torgovitsky}{Torgovitsky}{2015}]{torgovitsky2015identification}
\textsc{Torgovitsky, A.} (2015): \enquote{Identification of nonseparable models
  using instruments with small support,} \emph{Econometrica}, 83, 1185--1197.

\bibitem[\protect\citeauthoryear{Torgovitsky}{Torgovitsky}{2017}]{torgovitsky2017minimum}
---\hspace{-.1pt}---\hspace{-.1pt}--- (2017): \enquote{Minimum distance from
  independence estimation of nonseparable instrumental variables models,}
  \emph{Journal of Econometrics}, 199, 35--48.

\bibitem[\protect\citeauthoryear{van~der Vaart and Wellner}{van~der Vaart and
  Wellner}{1996}]{van1996weak}
\textsc{van~der Vaart, A. and J.~Wellner} (1996): \emph{Weak Convergence and
  Empirical Processes: With Applications to Statistics}, Springer Science \&
  Business Media.

\end{thebibliography}

\end{document}